\documentclass[twoside,11pt]{article}
\usepackage{jmlr2e}
\usepackage{amssymb}
\usepackage{amsmath}
\newtheorem{assum}[theorem]{Assumption} 

\makeatletter
\let\ftype@table\ftype@figure
\makeatother

\ShortHeadings{Likelihood-free  Model Choice with   the Jensen--Shannon Divergence}{Corander, Remes, and Koski}
\firstpageno{1}

\begin{document}

\title{Likelihood-free  Model Choice for Simulator-based  Models   with   the Jensen--Shannon Divergence}
\author{\name Jukka Corander  \email jukka.corander@medisin.uio.no \\ 
\addr  Department of Mathematics and Statistics and Helsinki Institute of Information Technology (HIIT) \\  University of Helsinki \\ 
  Pietari Kalmin katu 5, 00014 Helsingin Yliopisto, Finland \\
\addr  Department of Biostatistics, Institute of Basic Medical Sciences \\ University of Oslo \\
Sognsvannsveien 9, 0372 Oslo, Norway \\
 Parasites and Microbes, Wellcome Sanger Institute \\ Cambridge, CB10 1SA, UK \AND   \name Ulpu Remes 
\email  u.m.v.remes@medisin.uio.no\\  \addr  Department of Biostatistics, Institute of Basic Medical Sciences \\
University in Oslo \\ Sognsvannsveien 9, 0372 Oslo, Norway 
\AND  \name  Timo Koski \email tjtkoski@kth.se\\ \addr  Department of Mathematics and Statistics and Helsinki Institute of Information Technology (HIIT)\\ University of Helsinki \\
 Pietari Kalmin katu 5, 00014 Helsingin Yliopisto, Finland \\
\addr  KTH Royal Institute of Technology \\ Lindstedtsv\"agen 25, 100 44 Stockholm, Sweden}

\editor{NN}

\maketitle

\begin{abstract}
Choice of appropriate structure and parametric dimension of a model in the light of data has a rich history in statistical research, where the first seminal approaches were developed in 1970s, such as the Akaike's and Schwarz's model scoring criteria that were inspired by information theory and embodied the rationale called Occam's razor. After those pioneering works, model choice was quickly established as its own field of research, gaining considerable attention in both computer science and statistics. However, to date, there have been limited attempts to derive scoring criteria for simulator-based models lacking a likelihood expression. Bayes factors have been considered for such models, but arguments have been put both for and against use of them and around issues related to their consistency. Here we use the asymptotic properties of Jensen--Shannon divergence (JSD) to derive a consistent model scoring criterion for the likelihood-free setting called JSD-Razor. Relationships of JSD-Razor with established scoring criteria for the likelihood-based approach are analyzed and we demonstrate the favorable properties of our criterion using both synthetic and real modeling examples.   
\end{abstract}

\section{Introduction}
\label{ppo:sec:intro}

A research field of increasing popularity deals with simulator-based models that lack an expression for the likelihood of data and consequently require likelihood-free inference approaches to be used for model fitting; for a recent comprehensive overview see  \citet{cranmer}. In one of the pioneering works dealing with likelihood-free inference, such a model and its  likelihood function were called  \textit{implicit}  \citet{diggle1984monte}. In contrast, a model with explicit likelihood is a \textit{prescribed} statistical model. The different  simulator-based models share the basic idea to  adjust  the  parameters   by finding   values which yield  outputs   that resemble the observed data, which raises the issue  of  the assessment of the discrepancy between the observed and simulated data. Considerable advances have been made in how such discrepancy can be converted into approximate likelihood or used to obtain samples from the corresponding posterior distribution of model parameters, however, the question of how to appropriately adjust the discrepancy for changes in model complexity/dimension have been given much less attention in likelihood-free inference research. Largely only the use of Bayes factors in the context of Approximate Bayesian Computation (ABC) have obtained a serious consideration, see  \cite {beaumont2019approximate}, \citet{didelot2011likelihood},  \citet{leuenberger2010bayesian}, \citet{marin2018likelihood}   and  \citet{robert2007lack}.

In related work \citet{corremkos},  we studied the discrepancy   between the observed and simulated data  as measured by the (symmetric) Jensen--Shannon divergence (JSD). It was shown that the asymptotic properties of JSD can be succinctly used to derive estimators, confidence intervals and hypothesis tests for implicit models with categorical output distributions. Here we develop the theory further to obtain an information-theoretically inspired model scoring criterion for such implicit models that can be used for solving the model choice problem in a consistent manner.   
We restrict  to   simulator-based models  that  emit  categorical data, as such data has been the main field  of application of the method, see \citep{corander2017frequency}.   

Our method of model choice is a modification of  Occam$^{,}$s  Razor as developed  by Balasubramanian and co-workers   in  \citet{balasubramanian2005mdl}, \citet{balasubramanian1996geometric} and  
 \citet{myung2000counting}.   Occam$^{,}$s  Razor is based on  an intuitive geometric interpretation of the  meaning of complexity in  model selection.  To cite   \citet{balasubramanian1996geometric},  complexity measures the ratio of the volume  occupied by distinguishable distributions in a model that come close   to the truth relative to the volume of the model as a whole.   Briefly stated, our   modification is to  replace  the Kullback Leibler divergence  in  Occam$^{,}$s  Razor with the Jensen--Shannon divergence  and correspondingly call the result a  JSD-Razor.  Minimization of  $- 1\cdot $logarithm of  JSD-Razor   leads to a criterion that can be used to ranking simulator-based models with respect to the fidelity of their simulation outputs, such that the complexity of the model is accounted for. Asymptotic analysis of the logarithm of  JSD-Razor  leads to two criteria   for model choice, where  the observed fit of a model  is in terms of  minimized JSD   additively penalized. There are two expressions of penalty, the more subtle  one  is  accounting  for  the complexity of the models in the sense of  the geometric interpretation, but is not readily computable.

Model choice between  prescribed models has been extensively studied by    a number  different approaches, see,  e.g., \citet{atkinson1970method}, 
\citet{massart2000some},  or the survey   in  \citet{rao2001model}.  A choice between models
 based on a measure that indicates  the relative flexibility of the models  examining  the extent to which the
candidate models can mimic each other is studied  in \citet{wagenmakers2004assessing}.  This mimicry is based on bootstrapping both  observed  and  
simulated data from prescribed models.  Such a bootstrapping strategy could have been an option for  this work, too. 

The  minimum description length, see  \cite {roos2017minimum}, and Bayesian approaches, see \citet{cavanaugh1999generalizing}, to model choice  are broadly speaking  dealing with the observed fit of a model additively penalized by terms accounting  for  the complexity of the models.  The  model  choice  by  stochastic complexity incorporated in the normalized maximum  likelihood estimate was 
developed in  \citet{rissanen2007information}. The computational techniques of the normalization, when dealing with  nonparametric   models  for categorical  data are found  in \citet{kontkanen2007linear} and  \citet{mononen2008computing}. When the observed fit is measured in terms of the maximized likelihood  function,  this is not feasible  in simulator-based modeling.  

JSD is an instance of a $\phi$-divergence, see, e.g.,  \citet{osterreicher2002csiszar} for a survey.    
The work in   \citet{alba2020model} deals with model choice  on misspecified  prescribed models for categorical data using a general $\phi$-divergence  for fit  and an additive  penalty for  empty cells and is  fundamentally different from the piece of work here.

\section{ Simulator-based Models for Categorical Data}\label{sec:phiestimat11} 
In this  section  a  set of  definitions and notation is recapitulated for  probability distributions for categorical data. This 
involves naturally the probability simplexes in  Euclidean spaces.  
The notion of implicit statistical models for categorical data is defined  formally.  This introduces parameters $\theta$ in the formalism.  Thereafter  one can discuss the various 
settings  for  model choice: separate, overlapping and nested parameter spaces.  Finally we present the so called Birch conditions for  categorical probability distributions with parameter dependencies.  These conditions  lead to  the existence of  the maximum likelihood estimate in the implicit model. 
\subsection{Categorical Distributions }\label{sec:phiestimat} 
Let  ${\cal A}=\{a_{1},\ldots, a_{k}\}$ be a finite set, $k \geq 2$. We are concerned with a situation where 
$k$  and  all  categories $a_{j}$ are known. This excludes the issues   of  very large alphabets discussed in \citet{kelly2012classification}.     $\mathbb{R}^{{\cal A}}$  denotes   the set of real valued functions on  ${\cal A}$. 
We introduce the set of categorical (probability)  distributions as 
\begin{equation}\label{allprob}
\mathbb{P} = \left\{\text{all probability distributions on ${\cal A}$} \right\} \subset  \mathbb{R}^{{\cal A}}.
\end{equation}
The   Iverson bracket  $I_{i}(x )  = [x=a_i]\in  \mathbb{R}^{{\cal A}}$ is  defined for each $a_{i} \in {\cal A}$ by 
\begin{equation}\label{iversonbr}
 I_{i}(x )  =  [x=a_i]:=\left\{ \begin{array}{cc}  1 &  x =a_{i} \\
0 &x  \neq  a_{i} \end{array} \right.
\end{equation}
Any  $ P \in \mathbb{P} $ can be written   as   
\begin{equation}\label{iverson}
P(x)= \prod_{i=1}^{k} p_{i}^{[x=a_i]}, x \in {\cal A},
\end{equation}
where  ($0^{0}=1, 0^{1}=0$), and $ p_{i} \geq 0$, $\sum_{i=1}^{k}p_{i}=1$. 
 The support 
of $P \in \mathbb{P}$  is  
$
{\rm supp}(P)= \{ a_{i} \in {\cal A}  |  \quad p_{i}=P(a_{i}) >0  \}
$.  If  $X$ is a random variable (r.v.) assuming values on ${\cal A}$, 
$X \sim P$ $ \in \mathbb{P}$  means that   $P(X=x)= P(x)$ for all $x \in {\cal A}$.  

Any  $P \in \mathbb{P}$  is also identified  as a probability vector ${\bf p}$,  an element of the probability simplex  $\triangle_{k-1}$ defined by  
\begin{equation}\label{simplex}
\triangle_{k-1} :=\left \{ {\bf p} = \left(p_{1}, \ldots, p_{k} \right)\mid  p_{i} \geq 0 , i=1, \ldots,k;   \sum_{i=1}^{k}p_{i}=1 \right \} \subset \mathbf{R}^{k}.
\end{equation}  
We write  this one-to-one correspondence between $\mathbb{P}$ and $\triangle_{k-1} $ as 
\begin{equation}\label{trmap}
\triangle\left(P \right)= {\bf p}. 
\end{equation} 
The  $ i$-th face of $\triangle_{k-1}$  is defined as    $\partial_{i}\triangle_{k-1}=\{ {\bf p} \in \triangle_{k-1} | p_{i}=0 \}$.  Any  face is in fact a probability 
simplex in $\mathbf{R}^{k-1}$.
The simplicial boundary of  $\triangle_{k-1}$  is  $\partial \triangle_{k-1} =\cup_{i=1}^{p}\partial_{i}\triangle_{k-1}$ $= \{ {\bf p} \in \triangle_{k-1}  |p_{i} =0 
\quad  \text{for some $i$} \}$.
The simplicial or topological  interior of  $ \triangle_{k-1} $   is $\stackrel{o}{\triangle}_{k-1} := \triangle_{k-1} \setminus \partial \triangle_{k-1} $, i.e., 
\begin{equation}\label{topint}
\stackrel{o}{\triangle}_{k-1}= \{ {\bf p} \in \triangle_{k-1}  |p_{i} > 0,  i=1, \ldots,k;   \sum_{i=1}^{k}p_{i}=1  \}.
\end{equation}
We note that 
\begin{equation}\label{fsupport} 
{\rm supp}(P)=  {\cal A}   \Leftrightarrow \triangle(P) \in   \stackrel{o}{\triangle}_{k-1}.
\end{equation}
The assumption  
\begin{equation}\label{posass}
 \triangle(P) \in   \stackrel{o}{\triangle}_{k-1}
\end{equation}
is made  for all models in the sequel. 

Let us set $ z^{\ast}_{i} =\ln \frac{p_{i}}{p_{k}}$, and $k^{\ast}({\bf z}^{\ast}) = -  \ln p_{k}$. Then  any categorical distribution in $ \mathbb{P}$  in Equation~(\ref{iverson})   can be written as
$$
P_{{\bf z}^{\ast}}(x)=  e^{\sum_{i=1}^{k}z^{\ast}_{i}[x=a_{i}] -k^{\ast}({\bf z}^{\ast})}, \quad  x \in {\cal A},
$$
which is an exponential family of distributions.  It is shown in \citet[p.~186]{amari2010information} that ${\bf p}$ and  ${\bf z}^{\ast}$ are, respectively,  the primal and dual variables in the entropy geometry of   
$\mathbb{P}$.  In this  $\mathbb{P}$ is a Riemann manifold, where the squared local distance is determined by the Hessian of  $ G({\bf p})=\sum_{i=1}^{k} \left( p_{i} \ln p_{i}-p_{i}\right)$. \citet{amari2010information}  as well as \citet{pistone2019lecture} 
deal with the information geometry of the non-parametric  probability simplex, not the parameterized  ones  of  Section \ref{sec:phiestimat22}.
An  argument for  indexing  probability simplices  with parameters    in terms of statistical  precision   is found in   \citet{altham1984improving}.

\subsection{Simulator Modeling Represented as an Implicit Statistical  Model}\label{sec:phiestimat22} 

Consider  $P_{o} \in \mathbb{P}$  as   the so-called  true distribution.  $P_{o}$ is  otherwise unknown, except the  observed data 
$\mathbf{D}= (D_{1}, \ldots, D_{n_{o}}) $,  are assumed to be  an i.i.d. $n_{o}$-sample  from a data source under $P_{o}$,  $\mathbf{D} 
\sim P_{o}$. In likelihood-free inference one as a rule  reduces the observed data $\mathbf{D}$  to
some features, or summary statistics, before performing inference. The role  of the
summary statistics is to reduce the dimensionality and to filter out information which is not
deemed relevant for the inference. 
 The summary statistics in this work  will be the empirical distribution $\widehat{P}_{\mathbf{D}} \in \mathbb{P}$. This is  computed  
in terms of the relative frequencies of the categories $a_{j}$ in   $\mathbf{D}$. Formally, we write           
$
\widehat{p}_{i}= \frac{n_{i}}{n_{o}}, \quad i=1, \ldots, k$,
where $n_{i} = $ the number of samples $Z_{j}$ in $\mathbf{D}$ such that $Z_{j}=a_{i}$, and  following Equation~(\ref{iverson})
\begin{equation}\label{typen}
\widehat{P}_\mathbf{D}(x)=\prod_{i=1}^{k} \widehat{p}_{i}^{[x=a_{i}]}, \quad x \in{\cal A}. 
\end{equation}
The sufficiency of this summary statistics has been established in \cite{corremkos2}. 

Nect, ${\mathbf M}_{C}$  is  a  simulator model for the data source.  Citing   \citet{lintusaari2017fundamentals}, simulator models   can be understood  in our case as computer programs  that  take as input  random numbers  $V$ and the parameter $\theta  \in \Theta \subset  \mathbb{R}^{d}$, $d ={\rm dim}\left(\Theta \right) < k =\mid{\cal A}\mid$  and  produce as output  $\mathbf{X}=\left(X_{1}, \ldots, X_{n}\right)$,  $n$ i.i.d. samples of categories in  ${\cal A}$.   The platform in \citet{lintusaari2017fundamentals} is not limited to  categorical data. 
 We write the corresponding function as 
 ${\mathbf M}_{C}(\theta)$.   
By this designation,
$
p_{i}(\theta)= P \left( X= a_{i}  \right) = P \left( {\mathbf M}_{C}(\theta) = a_{i}  \right) 
$
 for any  $\theta \in \Theta$  induces the category probabilities, i.e.,   $k$ functions that  have no (fully)  explicit expression, i.e., they are implicit functions  of $\theta$ in the sense  of \citet{diggle1984monte}   satisfying  $p_{i}(\theta) \geq 0$, $\sum_{i=1}^{k} p_{i}(\theta)=1$ for all $\theta \in \Theta$.   There is the  distribution $P_{\theta} \in \mathbb{P}$  given by 
\begin{equation}\label{pariversion}
P_{\theta}(x):= \prod_{i=1}^{k}p_{i}(\theta)^{[x=a_{i}]}, x \in {\cal A}, \theta \in \Theta.   
\end{equation}   
The   implicit model  representation  of 
 ${\mathbf M}_{C}$ in $\mathbb{P}$  is denoted  by   $\mathbb{M}_{p}= \left\{ P_{\theta} \mid  \theta \in \Theta \right \}$, 
$$
 {\mathbf M}_{C} = \left\{ {\mathbf M}_{C} (\theta) \mid  \theta \in \Theta \right \}\models  \mathbb{M}_{p}= \left\{ P_{\theta} \mid  \theta \in \Theta \right \} \subset \mathbb{P}. 
$$
We are going to use the customary notation   $ {\mathbf X}\sim   P_{\theta}$ which is to be understood in the above sense of generative simulator-based sampling, not as sampling  from  a known categorical distribution. $\widehat{P}_{\mathbf{X}} \in \mathbb{M}_{n}(\theta) $ is a shorthand for the statement that 
  $\widehat{P}_{\theta}$  is the  summary statistics for  $\mathbf{X} \sim   P_{\theta}$.
By Equation~(\ref{trmap})  there corresponds   to the representation  $\mathbb{M}_{p}= \left\{ P_{\theta} \in  \mathbb{P}\mid  \theta \in \Theta \right \} \subset \mathbb{P}$  a submanifold  of  $\triangle_{k-1}$  
by 
\begin{equation}\label{modellimage}
\triangle\left( \mathbb{M}_{p} \right) =\left \{  {\bf p} \in \triangle_{k-1} |  
\text{ $\exists  P_{\theta} \in \mathbb{M}_{p}$ such that $ {\bf p}= \triangle \left(P_{\theta} \right)$}  \right\}. 
\end{equation}
 Let us consider   simulator models  $\mathbb{M}_{p}^{(l)}$, $l=1,  \ldots, L$ with corresponding $\Theta^{(l)}$ $\subset \mathbb{R}^{k}$. 
 There are three different situations  for   $ \mathbb{M}_{p}^{(l_{1})}$  
and  $ \mathbb{M}_{p}^{(l_{2})}$  to be  related to each other. We assume the (weak)  identifiability  of the models, 
$
\theta \neq \theta^{'}  \Rightarrow  P_{\theta} \neq P_{\theta^{'}}, 
$
see Section \ref{antaganden}.
\begin{enumerate}
\item  $\Theta^{(l_{1})}$ and  $\Theta^{(l_{2})}$ are separate, i.e.,  $\Theta^{(l_{1})} \bigcap \Theta^{(l_{2})} =  \emptyset$.  This can also mean that $\Theta^{(l_{1})}$ and  $\Theta^{(l_{2})}$  have  no relations between
each other, this is the problem of  model choice  in \citet{atkinson1970method}.
\item   $\Theta^{(l_{1})}$ and   $\Theta^{(l_{2})}$ are overlapping, i.e. $\Theta^{(l_{1})}$ and   $\Theta^{(l_{2})}$ have a nonempty intersection, but 
are not subsets of each other. 
\item   $\Theta^{(l_{1})}$ and   $\Theta^{(l_{2})}$ are nested,  e.g.,   $\Theta^{(l_{1})}  \subset \Theta^{(l_{2})}$,
\end{enumerate}
see, e.g., \citet[p.~317, p.~320, p.~323]{vuong1989likelihood}, 
where   the submanifolds $ \mathbb{M}_{p}^{(l_{i})}$ are  for prescribed models. 
\begin{example}\label{multicateg}
Suppose that  each  category $a_{i}$   is associated with  a predictor $ \mathbf{\alpha}^{(i)}$ with $d$  state variables $\alpha^{(i)}_{s} $, which may  be real-valued, binary-valued, categorical-valued, etc.,  fixed   characteristics    of  the category  $a_{i}$. Let us write 
\begin{eqnarray}
a_{i}  & \leftrightarrow & \mathbf{\alpha}^{(i)}= \left( \alpha^{(i)}_{1}, \ldots,\alpha^{(i)}_{d} \right), i=1, \ldots , k-1 \nonumber \\
a_{k}  & \leftrightarrow &  \mathbf{\alpha}^{(k)} = \mathbf{0}_{d}:= \left( 0, \ldots,  0 \right) \nonumber 
\end{eqnarray}
with the arbitrary choice  of  $a_k$ as  a base   of $d$ zero states.  $\theta= \left(\theta_{1}, \ldots,\theta_{d} \right)$ 
is parameter vector  in some $\Theta$.     We set 
$\langle \mathbf{\alpha}^{(i)}, \theta \rangle :=\sum_{s=1}^{d}\alpha^{(i)}_{s}  \theta_{s}  $. Furthermore we set 
$
M(\theta):= \ln \left(  1 + \sum_{i=1}^{k-1} e^{\langle \mathbf{\alpha}^{(i)}, \theta \rangle }\right)
$ 
and the (prescribed) category probabilities in Equation~(\ref{pariversion}) are defined by 
\begin{equation}\label{multicat1}
p_{i}(\theta):= e^{\langle \mathbf{\alpha}^{(i)}, \theta \rangle -M(\theta)}, \quad i=1, \ldots, k, 
\end{equation}  
where the convention on  $\alpha^{(k)}$ gives 
$
p_{k}(\theta)= 1/\left(  1 + \sum_{i=1}^{k-1} e^{\langle \mathbf{\alpha}^{(i)}, \theta \rangle} \right) 
$.
Here we can obviously define  nested models with different dimensions by identifying     $\theta^{l_{1}}= \left(\theta_{1}, \ldots,\theta_{l_{1}} \right)$  
  $\in \Theta^{(l_{1})} $   as  $ \theta^{l_{2}}= \left(\theta_{1}, \ldots,\theta_{l_{1}},  0, \ldots, 0 \right)$ $\in \Theta^{(l_{2})}$, where  
$l_{1} < l_{2} < k$.  The  parameter vector   $\theta=\mathbf{0}_{l} $ with appropriate number of zeros  lies thus  in  every  $\Theta^{(l)}$, and Equation~(\ref{multicat1})
becomes the discrete uniform distribution $P_{U}$ on ${\cal A}$, i.e.,  
\begin{equation}\label{unif}
P_{U}(x)= \prod_{i=1}^{k} \left(  \frac{1}{k}\right)^{[x=a_i]},  \quad  x \in {\cal A}. 
\end{equation}
 Hence $P_{U}$   can be regarded as a model with dimension zero, as  $k$ is known in advance.    $\triangle\left( P_{U} \right)$ is known as   the barycenter of  $\triangle_{k-1}$. \end{example}

\begin{example}\label{multicateg2}
We consider a  special case of Example \ref{multicateg} with  $k=3$  and $d=0,1,2$.  Suppose that  each  category $a_{i}$   is associated with  a two-bit  string as follows 
$$
a_{1}   \leftrightarrow  \mathbf{\alpha}^{(1)}= \left( 1 ,  0\right) , 
a_{2}  \leftrightarrow  \mathbf{\alpha}^{(2)}= \left( 0 , 1\right), 
a_{3}   \leftrightarrow  \mathbf{\alpha}^{(3)}= \left( 0 , 0\right) .
$$
The nested models are given in terms of  $\theta \in \mathbb{R}^{2}$. 
\begin{description}
\item[ (i) $\theta= \left(\theta_{1},  \theta_{2} \right)$]
 Substitution  in   Equation~(\ref {multicat1}) gives 
$
M_{2}(\theta)= \ln \left(  1 +  e^{ \theta_{1}}   + e^{ \theta_{2}}   \right)
$  and 
\begin{equation}\label{multicat12}
p_{1}(\theta)= e^{ \theta_{1}    -M_{2}(\theta)}, p_{2}(\theta)= e^{ \theta_{2}    -M_{2}(\theta)},  p_{3}(\theta)= \frac{1}{ 1 +  e^{ \theta_{1} }  + e^{ \theta_{2}}  }.
\end{equation} 
\item[ (ii) $\theta= \left(\theta_{1}, 0\right)$] 
Here $M_{1}(\theta)= \ln \left( 2 +  e^{ \theta_{1}}      \right)$
and  Equation~(\ref {multicat1}) gives  
\begin{equation}\label{multicat13}
p_{1}(\theta)= e^{ \theta_{1}    -M_{1}(\theta)}, p_{2}(\theta)= \frac{1}{ 2 +  e^{ \theta_{1} } } ,  p_{3}(\theta)= \frac{1}{ 2 +  e^{ \theta_{1} } }  .
\end{equation} 
\item[ (iii) $\theta= \left(0, 0\right)$] 
And  $M_{0}(\theta)= \ln \left(  3     \right)$,  Equation~(\ref {multicat1}) gives 
\begin{equation}\label{multicat14}
p_{1}(\theta)= \frac{1}{ 3 }, p_{2}(\theta)=\frac{1}{ 3 },  p_{3}(\theta)= \frac{1}{ 3 }.
\end{equation} 
\end{description}
\end{example}

\subsection{Assumptions and Existence of Maximum Likelihood Estimate for Simulator Modeling} \label{antaganden}
For the further analysis  a set of  notations and assumptions on the flexibility of the simulator model  are required.
This is analogous  to the  KOH  theory of the smoothness of the functions in $\mathbf{M}_{C}$ \citep[p.~2]{kennedy2000predicting}. 
We need some notational conventions. In the sequel  $  \mathbf{x} \in  \mathbf{R}^{k}$ is a $1 \times k$, a  row vector, and $  \theta  \in \Theta \subset  \mathbf{R}^{d}$ is a $1 \times d$   row vector.  Hence 
${\bf x} {\bf x}^{T}=\sum_{i=1}^{k} x_{i}^{2}$ is a scalar product.   $|| \mathbf{x} ||_{2, \mathbf{R}^{k}} =\sqrt{{\bf x} {\bf x}^{T}}$ is  the Euclidean norm on $ \mathbf{R}^{k}$ and similarly  for 
 $|| \theta ||_{2, \mathbf{R}^{d}}$.
\begin{assum}\label{kontdiff}
 For every $\theta_{o}$ in the interior of $\Theta$  and every $j=1, \ldots, k$ we  have 
\begin{equation}\label{nablaexp1}
 p_{j}(\theta)=   p_{j}(\theta_{o})+ \left( \theta - \theta_{o} \right)p_{j}^{'}(\theta_{o})^{T} + o\left( || \theta - \theta_{o} ||_{2, \mathbf{R}^{d}}\right).
\end{equation}
where we have the $1 \times d$   total differential  
\begin{equation}\label{nablanot4}
p_{j}^{'}(\theta):=  \left( \frac{\partial}{\partial \theta_{1}}p_{j}(\theta), \ldots, \frac{\partial}{\partial \theta_{d}}p_{j}(\theta) \right).
\end{equation}
 \end{assum}
Let  us define for $i \in \{1, \ldots, d\}$ and $j \in  \{1, \ldots, d\}$  and $X \sim P_{\theta}$  the expectation 
\begin{equation}\label{ijfisher}
I_{ij}(\theta) := E \left[  \frac{\partial}{\partial \theta_{i}}  \ln P_{\theta}(X) \frac{\partial}{\partial \theta_{j}}  \ln P_{\theta}(X) \right].
\end{equation}  
The $d \times d$ matrix  
\begin{equation}\label{nablanot2}
 I(\theta):= \left[ I_{ij}(\theta) \right]_{i=1,j=1}^{d,d}
\end{equation}
is  the Fisher   information matrix of $\mathbb{M}_{p}$ at $\theta$.
\begin{assum}\label{kontdiff2}
For $\theta_{o}$ such that $p_{j}(\theta_{o}) >0$  for each $j=1, \ldots, k$  and  for all      $\theta$ in the interior of $\Theta$
\begin{equation}\label{nablaexp2}
 p_{j}(\theta)=   p_{j}(\theta_{o})+ \left( \theta - \theta_{o} \right)p_{j}^{'}(\theta_{o})^{T} +  \frac{1}{2}(\theta - \theta_{o})H(\theta_{o})  ( \theta - \theta_{o})^{T} +  o\left( || \theta - \theta_{o} ||^{2}_{2, \mathbf{R}^{d}}\right)
\end{equation}
where  $H(\theta_{o}) $  is  the Hessian with the elements 
\begin{equation}\label{ijfisher34}
H_{ij}(\theta_{o}) = -E  \ \left[  \frac{\partial^{2}}{\partial \theta_{i}\partial\theta_{j}}  \ln P_{\theta_{o}}(X) \right].
\end{equation}  
Under further regularity assumptions we have   $ H(\theta_{o}) =-I(\theta_{o})$. 
 \end{assum}
The  Jacobian $J\left(\theta \right)$ is the $k \times d$ matrix  
\begin{equation}\label{nablanot3}
J\left(\theta \right):=  \left( \begin{array}{cc}    p_{1}^{'}(\theta) \\
\vdots \\
  p_{k}^{'}(\theta) \end{array} \right)
\end{equation}
which has  the vectors $p_{j}^{'}(\theta)$ defined in Equation~(\ref{nablaexp2}) as its  rows. Furthermore we consider 
the $k \times k$ diagonal matrix 
\begin{equation}\label{nabladiag}
\Lambda(\theta):= {\rm diag}\left(   \frac{1}{\sqrt{p_{1}(\theta)}},\ldots,   \frac{1}{\sqrt{p_{k}(\theta)}}  \right). 
\end{equation}
Let us define
\begin{equation}\label{afaktor}
A(\theta) :=  \Lambda(\theta)J(\theta),
\end{equation}
which is a $k \times d$-matrix  and assume 
\begin{assum}\label{invfisher}
The  rank of $A(\theta)$ is $d$. 
\end{assum}
The following assumption is the strong  identifiability  condition of    \citet[(B), p.~817]{birch1964}.
\begin{assum}\label{invkontbirch}
For any $\epsilon >0$ there exists $\delta >0$ such that, 
\begin{equation}\label{birchinvers}
\text{if $    || {\bf p}(\theta)- {\bf p}(\theta_{0})||_{2, \mathbf{R}^{k}} > \epsilon$,  then 
 $||\theta - \theta_{0}||_{2, \mathbf{R}^{d}} > \delta$.}
\end{equation}
\end{assum}
Clearly this implies the weak  identifiability  assumption
\begin{assum}\label{identifiabililty}
 \begin{equation}\label{identass} 
\theta \neq \theta^{'}  \Rightarrow  P_{\theta} \neq P_{\theta^{'}}. 
\end{equation} 
 \end{assum}
Under this assumption  $ {\bf p}= \triangle \left(P_{\theta} \right)$ is  a one-to-one  map between $\triangle\left( \mathbb{M}_{p} \right) $ and $\mathbb{M}_{p}$. 

\citet{birch1964} 
proved under these assumptions that the MLE, maximum (prescribed) likelihood estimate of a true distribution   $P_{\theta_{o}}$ in a $\mathbb{M}_{p}$   exists and is almost surely consistent.  For simulator-modeling we  re-formulate this by the next proposition. 
\begin{proposition}\label{existensmle}
 MLE exists and is consistent for the simulator model ${\mathbf M}_{C} = \left\{ {\mathbf M}_{C} (\theta) \right\}$, in the sense that  MLE  exists and is consistent  for  $ \mathbb{M}_{p}= \left\{ P_{\theta} \mid  \theta \in \Theta \right \}$ for compact $\Theta$  under Assumptions \ref{kontdiff}--\ref{invkontbirch}, where  
${\mathbf M}_{C} = \left\{ {\mathbf M}_{C} (\theta) \mid  \theta \in \Theta \right \}$ $ \models  \mathbb{M}_{p}=\left\{ P_{\theta} \mid  \theta \in \Theta \right \}$.
\end{proposition} 
The claim about existence is established  in  \citet[p.~818]{birch1964} also for general $\Theta$  by various  considerations of how MLE can be introduced. It has been shown in \citet{corremkos}  for categorical data that the minimum JSD estimate, to be defined  in Equation~(\ref{minjsd}), and 
MLE are    asymptotically equal, when  $n_{o} \rightarrow +\infty$.  It is shown in \citet{corremkos2} that the minimum JSD estimate exists and is measurable   for compact  $\Theta$.

\section{The Symmetric Jensen--Shannon Divergence }
In this section we introduce  the (symmetric) Jensen--Shannon divergence (JSD)  as the measure of discrepancy between the observed and simulated  data summaries. We  note an interpretation of  JSD  as redundancy of  the optimal source coding for a mixture source.  We establish the JSD as a $\phi$-divergence, which makes it possible to 
use certain general properties of    $\phi$-divergences for JSD.    
\subsection{ Definition of Symmetric  JSD } 
Consider  two generic categorical probability   distributions $P\in \mathbb{P}$: $ P(x)= \prod_{i=1}^{k} p_{i}^{[x=a_{i}]}$  and $Q\in \mathbb{P}$: $ Q(x)= \prod_{i=1}^{k} q_{i}^{[x=a_{i}]}$. Then 
\begin{equation}\label{klinformation}
 D_{\rm KL}( P, Q) := \sum_{x \in {\cal A} } P(x) \ln \left( \frac{ P(x)}{ Q(x)}\right) = \sum_{i=1}^{k} p_{i} \ln \left( \frac{ p_{i}}{q_{i}}\right) 
 \end{equation}
is known as  the Kullback--Leibler divergence (KLD). In general  $D_{\rm KL}( P, Q) \neq  D_{\rm KL}( Q,P)$, if $P\neq Q$. We use $0\ln0 =0$ and if  ${\rm supp}(Q) \subset {\rm supp}(P)$,  we take $D_{\rm KL}( P, Q) = +\infty$. Next,  the  symmetric Jensen--Shannon divergence is   denoted by  $D_{\rm JS}( P, Q)$,    and  is defined with $M:= \frac{1}{2} P + \frac{1}{2}Q$,
i.e.,  $M(x)= \prod_{i=1}^{k} \left( \frac{1}{2} p_{i} +  \frac{1}{2} q_{i} \right)^{[x=a_{i}]}$, as 
  \begin{equation}\label{jssymm} 
D_{\rm JS}( P, Q) :=  \frac{1}{2}  D_{\rm KL}( P, M) +  \frac{1}{2}  D_{\rm KL}( Q,  M).
 \end{equation} 
 $D_{\rm JS}( P, Q)$  is  a  symmetrized version of KLD,  as $ D_{\rm JS}( P, Q)= D_{\rm JS}( Q,P)$ 
and  a smoothed version, since  $ D_{\rm JS}( P, Q) $ is uniformly bounded even if ${\rm supp}(Q) \subset {\rm supp}(P)$ or ${\rm supp}(P) \subset {\rm supp}(Q)$, as found in Equation~(\ref{range}).

The important result  in the  following  proposition  is  provided  for visibility  and ease of  reference.  
\begin{proposition}\label{sqrmetric}
$  \sqrt{D_{\rm JS}\left( P,Q \right)}$ is a metric on $\mathbb{P}\times \mathbb{P}$. 
\end{proposition} 
This is established  in \citet{endres2003new}, see also \citet{vajda2009metric}.  We shall also use $  D^{1/2}_{\rm JS}\left( P,Q \right)$   for $  \sqrt{D_{\rm JS}\left( P,Q \right)}$. 

The symmetric JSD  is often used in machine learning, see e.g.,  \citet{corander2017frequency} and   \citet{corremkos} for references.  There are   reasons for that: $D_{\rm JS}( P, Q)$ is bounded, nonnegative  and has an operational meaning pointed out next.  

\subsection{Interpretation as  Redundancy of Source Code for  $ X \sim \left( \frac{1}{2} P + \frac{1}{2} Q \right)$} \label{topsoecode}
We recapitulate  the idea from  \citet{topsoe1979information}, see also   \citet[p.~1858]{endres2003new}. 
As above,  consider  drawing  an i.i.d. sample $ \mathbf{X}= \left(X_{1}, \ldots, X_{n} \right)  \sim \left( \frac{1}{2} P + \frac{1}{2} Q \right)$,  
where $P$  and $Q$ are known distributions.  For any  $X_{i}$ we do not know which of $P$  or $Q$
was drawn from. Next we  seek  the  source  coding   that gives the
shortest average code length for the compression  of    $ \mathbf{X}$, see \citet[Ch. 5.3]{cover2012elements}. Let $R \in \mathbb{P}$  and  
$l_{i}$  be  the code length  $l_{i}=-  \ln r_{i}$. Let us call this code $\kappa$. Then the  expected code length of $\kappa$  is 
$$
\frac{1}{2}\sum_{i=1}^{k} l_{i}p_{i} + \frac{1}{2}\sum_{i=1}^{k} l_{i}q_{i}. 
$$
The  minimal  code length is obtained by selecting $R=M$  and the minimum is 
 the Shannon entropy of   $M $ in natural  logarithm, 
$$
H(M)=H\left( \frac{1}{2} P  + \frac{1}{2}Q \right):= -\sum_{i=1}^{k} \left(\frac{1}{2}p_{i} +  \frac{1}{2} q_{i}\right) 
  \ln \left(\frac{1}{2}p_{i} +  \frac{1}{2} q_{i}\right),
$$
see   \citet[Ch. 5.3]{cover2012elements}. On the other hand, a genie, who knows which of the two distributions  was chosen  
to generate the individual  $X_{i}$,   can by the same  argument  device a data compression  code with a shorter expected  minimum code length that is  
equal to $\frac{1}{2}H(P) + \frac{1}{2}H(Q) $. 
 Then    $D_{\rm JS}( P, Q)$ is the redundancy of the code  $\kappa$, because   
\begin{equation}\label{klinformationiudnet}
 D_{\rm JS}( P, Q) =H(M) -  \frac{1}{2}H(P) -  \frac{1}{2}H(Q).  
\end{equation}
This is a special case of an  identity  in \citet[Lemma 4]{topsoe1979information}.  The  right hand side of Equation~(\ref{klinformationiudnet})  is the Shannon-Jensen divergence  $D_{\rm JS}( P, Q)$  as  defined in \citet{lin1991divergence}.   

As is well-known, see \citet[p.~19]{cover2012elements},  $D_{\rm KL}(P,M)$  can be understood  as the
inefficiency of assuming that the true distribution is $M$ when it actually is $P$. Therefore  then  $ D_{\rm JS}( P, Q)$  could be seen as a minimum inefficiency distance, as formulated  in   \citet[p.~1859]{endres2003new}. 

\subsection{JSD and  $\phi$-divergences }
Consider 
\begin{equation}\label{jsddiv}
\phi_{\rm JS}(u)= \frac{1}{2} u \ln u - \frac{1}{2}\left( u +1 \right)\ln  \left(   \frac{1}{2} u +  \frac{1}{2} \right),  0 <  u < +\infty.
\end{equation}
One can  check that  $\phi_{\rm JS}$ is a  convex function on $(0, + \infty)$ $\stackrel{\phi}{\mapsto} \mathbf{R}$  and  has the properties $0 \phi_{\rm JS} \left( \frac{0}{0} \right) =0$ and $0\phi_{\rm JS}(x/0)= \lim_{\epsilon  \rightarrow 0} \epsilon
\phi_{\rm JS}(x/\epsilon)$,  $\phi(1)=0$.   It holds also  that 
\begin{equation}\label{jsddivdef}
D_{\rm JS}( P, Q) =  \sum_{x \in {\cal A} } Q(x) \phi_{\rm JS}\left( \frac{ P(x)}{ Q(x)}\right). 
\end{equation}
Hence  $D_{\rm JS}( P, Q) $ is a special case of  a $\phi$-divergence,  see   \citet[Ch. 8 and 9]{vajda1989theory} for the general theory and  \citet{osterreicher2002csiszar} for a concise summary.
In addition  the $^{\ast}$-conjugate $\phi^{\ast}$ of any divergence function  $\phi$  is defined by 
$
\phi^{\ast}(u):= u \phi \left( 1/u \right) $ for  $ 0 \leq u < +\infty$.
Then it turns out that    
\begin{equation}\label{range}
0=\phi_{\rm JS}(1) \leq D_{\rm JS}( P, Q)  \leq \phi_{\rm JS}(0) + \phi_{\rm JS}^{\ast}(0)= \ln (2). 
\end{equation}
 The left  equality implies the so-called identity of  of indiscernibles, i.e.,   $ D_{\rm JS}( P, Q)  =0$ if anly if 
$P=Q$.  The  inequalities  in Equation~(\ref{range}) are  an instance of  the range property   due to   Liese and Vajda in \citet[Thm 5]{liese2006divergences} valid for all $\phi$-divergences.   The study  \citet{topsoe2000some}  contains  several additional explicit expressions and  bounds, valid especially  for  $ D_{\rm JS}\left( P, Q \right)$. 

The following inequality  seems not be available in the literature, but is useful for our purposes.  
\begin{lemma}\label{lemmapinskerbd}
$P \in \mathbb{P}$, $Q \in \mathbb{P}$.  Then 
\begin{equation}\label{jsddkl}
 D_{\rm JS}( P, Q)  \leq \frac{1}{2}  D_{\rm KL}(P,Q). 
\end{equation}
\end{lemma}
\begin{proof} 
If  $  \frac{1}{2} D_{\rm KL}(P,Q) \geq  \ln 2$, the lemma holds by  Equation~(\ref{range}).  Otherwise, 
a small piece of  algebra applied on  Equations~(\ref{jsddiv}) and (\ref{jsddivdef})  reveals  that 
\begin{equation}\label{jsdkuusident}
D_{\rm JS}( P, Q) = \frac{1}{2} D_{\rm KL}( P, Q)   -D_{\rm KL}( M, Q).
\end{equation}
 Since  $D_{\rm KL}( M, Q) \geq 0$, the assertion follows.
\end{proof}

The analysis of the asympotics of  JSD and model choice in the sequel requires additionally   the introduction of  the $\phi$-divergence with  $\phi(x)=|x-1|$.   We obtain  
\begin{equation}\label{totvar}
V( P, Q): = \sum_{x \in {\cal A}} Q(x) \left | \frac{P(x)}{Q(x)}-1 \right | =   \sum_{i=1}^{k} | p_{i} - q_{i}|.  
 \end{equation}
V( P, Q) is  called the variation distance.  This is frequently discussed as the acceptance criterion  in ABC.

\subsection{Information Radius   and Model  Evidence  } \label{largedevhoeffd2}
Let $p(\theta)$ be a prior density on $\Theta$ and  
\begin{equation}\label{datallikelihood}
P\left(\mathbf{D}\mid \mathbb{M}_{p}\right ) := \int_{\Theta} P_{\theta} \left(\mathbf{D}\right)p(\theta)d\theta.   
\end{equation} 
is the marginal data likelihood, also  known as   the model evidence.  Here 
$d\theta$ is the Lebesgue measure induced on $\Theta$.
 
In information  theory   $\widehat{P}_{\mathbf{D}} $     is  called the type of $ {\mathbf D}$ on ${\cal A}$, see  
  \citet[Part I Ch. 2]{csiszar2011information} and   \citet[Ch.  11.1]{cover2012elements}. 
Let  $n=n_{o}$.   The type class of $ \widehat{P}_{\mathbf{D}}$  is defined, see   \citet[Ch.  11.1-11.3]{cover2012elements},  
  by  
\begin{equation}\label{typeclass}
{\cal T}_{n}\left(\widehat{P}_{\mathbf{D}}\right) := \{ {\mathbf X}= (X_{1}, \ldots,X_{n}) \in{\cal A}^{n} \mid  \widehat{P}_{\mathbf {X}}=  \widehat{P}_{\mathbf{D}} \}
\end{equation} 
The set of  all types on  ${\cal  A}$ for $n$ samples
\begin{equation}\label{typesone}
{\cal P}_{n} := \left\{  P \in  \mathbb{P} \mid  {\cal T}_{n}\left( P \right)  \neq \emptyset \right\}.
\end{equation}
Let 
\begin{equation}\label{tolerans}
A_{\epsilon}:= \left\{ {\bf X}  \in   {\cal A}^{n} \mid  D^{1/2}_{\rm JS}\left(\widehat{P}_{\mathbf{D}},  \widehat{P}_{\mathbf{ X}} \right)  \leq \epsilon \right\}.
\end{equation} 
 Let $p(\theta)$ be a prior density  on $\Theta$.  
\begin{proposition}\label{acceptancerate} Let  $n =n_{o}$. Then  
\begin{equation}\label{sllnh}
\lim_{\epsilon \downarrow 0}  \int_{\Theta} P_{\theta}^{(n)}\left( A_{\epsilon}\right)    p(\theta)  d \theta =\frac{n_{o}!}{\prod_{j=1}^{k} n_{o,j} !} 
P\left(\mathbf{D}\mid \mathbb{M}_{p}\right ). 
\end{equation}
\end{proposition}
\begin{proof} With the use of  Equation~(\ref{typeclass}) and Equation~(\ref{typesone})  we have 
$$
P_{\theta}^{(n)}\left(A_{\epsilon} \right)=\sum_{ {\bf X} \in A_{\epsilon}} P_{\theta}\left({\bf X} \right) =
 \sum_{ P \in  {\cal P}_{n} | D^{1/2}_{\rm JS}\left(\widehat{P}_{\mathbf{D}}, P \right) \leq  \epsilon} 
 \sum_{ {\bf X} \in  {\cal T}_{n}\left( P\right) }   P_{\theta}\left({\bf X} \right)
$$
From  \citet[Theorem 11.1.2]{cover2012elements}  we have  the identity 
\begin{equation}\label{typen111}
P_{\theta} (\mathbf{D}) =e^{ -n_{o}H\left(\widehat{P}_{\mathbf{D}}\right) -  n_{o}D_{\rm KL} \left( \widehat{P}_{\mathbf {D}}, P_{\theta}\right)},
\end{equation}
where  $H\left(\widehat{P}_{\mathbf{D}}\right) $  is the Shannon entropy of  $\widehat{P}_{\mathbf{D}}$  and thus  
$$
\sum_{ {\bf X} \in  {\cal T}_{n}\left( P\right) }   P_{\theta}\left({\bf X} \right)= 
\sum_{ {\bf X} \in  {\cal T}_{n}\left( P\right) } e^{ -nH\left(\widehat{P}_{\mathbf{X}}\right) - nD_{\rm KL} \left( \widehat{P}_{\mathbf {X}}, P_{\theta}\right)}. 
$$
However since  ${\bf X} \in  {\cal T}_{n}\left( P\right)$,   $\widehat{P}_{\mathbf{X}}=P$  by Equation~(\ref{typeclass})  and, with  
the cardinality  $ |  {\cal T}_{n}\left( P\right) | $, 
$$
\sum_{ {\bf X} \in  {\cal T}_{n}\left( P\right) }   P_{\theta}\left({\bf X} \right) =e^{ -nH\left(P\right) - nD_{\rm KL} \left( P , P_{\theta}\right)}|  {\cal T}_{n}\left( P\right) |
$$
and
$$
P_{\theta}^{(n)}\left( A_{\epsilon}\right)=
 \sum_{ P \in  {\cal P}_{n} | D^{1/2}_{\rm JS}\left(\widehat{P}_{\mathbf{D}}, P \right) \leq \epsilon} e^{ -nH\left(P\right) - nD_{\rm KL} \left( P , P_{\theta}\right)}|  {\cal T}_{n}\left( P\right) |. 
$$
The sum has a finite number of terms and hence the change of order between summation  and integration  is permitted, and 
we have 
$$
 \int_{\Theta} P_{\theta}^{(n)}\left( A_{\epsilon}\right)   p(\theta)  d \theta = 
 \sum_{ P \in  {\cal P}_{n} | D^{1/2}_{\rm JS}\left(\widehat{P}_{\mathbf{D}}, P \right) \leq \epsilon} |  {\cal T}_{n}\left( P\right) | \int_{\Theta}e^{ -nH\left(P\right) - nD_{\rm KL} \left( P , P_{\theta}\right)}d \theta.  
$$
As  $\epsilon \downarrow 0$,  the set of types $ \{ P \in  {\cal P}_{n} | D^{1/2}_{\rm JS}\left(\widehat{P}_{\mathbf{D}}, P \right) \leq \epsilon  \}$  decreases  to  $ \{ P \in  {\cal P}_{n} | D^{1/2}_{\rm JS}\left(\widehat{P}_{\mathbf{D}}, P \right) =0 \}$. Since  
  $ D^{1/2}_{\rm JS}\left(\widehat{P}_{\mathbf{D}}, P \right) =0$  if and only if  $\widehat{P}_{\mathbf{D}}=P$, and since 
$\widehat{P}_{\mathbf{D}} \in  {\cal P}_{n} $, as $n=n_{o}$, the limit set is the singleton set  $\left \{ \widehat{P}_{\mathbf{D}} \right \}$.  
Hence  the sum reduces in the limit to a single  term, or, to 
$$
\lim_{\epsilon \downarrow 0}  \int_{\Theta} P_{\theta}^{(n)}\left( A_{\epsilon}\right)   p(\theta)  d \theta=  
|  {\cal T}_{n}\left( \widehat{P}_{\mathbf{D}}  \right) |    \int_{\Theta} e^{ -nH\left(\widehat{P}_{\mathbf{D}} \right) -  nD_{\rm KL} \left(\widehat{P}_{\mathbf{D}} , P_{\theta}\right)}p(\theta) d\theta
$$ 
$$
= |  {\cal T}_{n}\left( \widehat{P}_{\mathbf{D}}  \right) |  \int_{\Theta} P_{\theta} \left(\mathbf{D}   \right) p(\theta)  d \theta,
$$
where   Equation~(\ref{typen111}) was used again.   Now,  by  combinatorics
$
|  {\cal T}_{n}\left( \widehat{P}_{\mathbf{D}}  \right) | =n_{o}!/\prod_{j=1}^{k} n_{o,j} !$.
By  Equation~(\ref{datallikelihood}) we have found the assertion in Equation~(\ref{sllnh}). 
\end{proof}
The proof above contains  the following special case, which shows  a scaled  $ P_{\theta}^{(n)}\left( A_{\epsilon}\right) $ as an asymptotically  correct estimate of the 
implicit  likelihood function. 
\begin{corollary}
\begin{equation}\label{sllnh2}
 \frac{\prod_{j=1}^{k} n_{o,j} !} {n_{o}!} \lim_{\epsilon \downarrow 0} P_{\theta}^{(n)}\left( A_{\epsilon}\right)  = P_{\theta}\left(\mathbf{D}\right ).
\end{equation}
\end{corollary}  
The integral $  \int_{\Theta} P_{\theta}^{(n)}\left( A_{\epsilon}\right)    p(\theta)  d \theta$ is known in the literature   as the acceptance rate of certain  ABC algorithms 
\citet[p.~247]{leuenberger2010bayesian}.
Section~\ref{jsdcritsugg} discusses a criterion of model choice that in view of Proposition \ref{acceptancerate}  will also  maximize the acceptance rate at  small  $\epsilon$.  

\subsection{Use in Parameter Estimation}

\newcommand{\param}{\theta}
\newcommand{\ncls}{k} 
\newcommand{\icls}{i} 
\newcommand{\nobs}{n_o} 
\newcommand{\nsms}{n} 
\newcommand{\nrep}{m} 
\newcommand{\pcls}{p_\icls}


The present work uses symmetric JSD as a model fit measure.
We assume an observed data set $\mathbf D$ summarized as $\widehat{P}_{\mathbf{D}}$ and a simulator-based model that produces categorical observation data.
Assuming that we can calculate category probabilities $P_\param$ based on the model parameters $\param$, 
we can find the parameters that maximize the model fit to observed data as
\begin{equation}\label{eq:minimum_jsd_estimate}
    \hat\param_\mathrm{JSD} = \arg\min_\param D_\mathrm{JS}(\widehat{P}_{\mathbf D}, P_\param).
\end{equation}
This is a special case of   the minimum $\phi$-divergence estimate  see,  e.g., 
\citet[Ch. 5.1--5.3]{pardo2018statistical}.   For  the minimum $\phi$-divergence estimate for discrete (incl. categorical) distributions, see \citet{morales1995asymptotic}. 
The parameter estimate $   \hat\param_\mathrm{JSD}$ is equivalent to the maximum likelihood estimate $\hat\param_\mathrm{ML}$, when the observation count $\nobs\rightarrow\infty$ \citep{corremkos}.


Simulator-based or likelihood-free estimation methods are needed when the mapping between model parameters and category probabilities is complicated or unknown so that a direct comparison between the observed data and model parameters is not possible.
The idea is that while we cannot calculate $D_\mathrm{JS}(\widehat{P}_{\mathbf D}, P_\param)$, we can simulate data with the model parameters $\param$ and evaluate model fit based on comparison between the observed and simulated data.
In practice, the individual simulations are used to calculate $D_\mathrm{JS}(\widehat{P}_{\mathbf D}, \widehat{P}_{\mathbf X_\param})$ and the optimization task is to find the model parameters that minimize the expected discrepancy $E[D_\mathrm{JS}(\widehat{P}_{\mathbf D}, \widehat{P}_{\mathbf X_\param})]$.


To summarize, when the mapping between between model parameters and category probabilities is unknown, we substitute $D_\mathrm{JS}(\widehat{P}_{\mathbf D}, P_\param)$ with $E[D_\mathrm{JS}(\widehat{P}_{\mathbf D}, \widehat{P}_{\mathbf X_\param})]$.
However, minimization  of $E[D_\mathrm{JS}(\widehat{P}_{\mathbf D}, \widehat{P}_{\mathbf X_\param})]$ is more difficult than minimizing $D_\mathrm{JS}(\widehat{P}_{\mathbf D}, P_\param)$ (which can be quite difficult, too).
We can show that in certain conditions $D_\mathrm{JS}(\widehat{P}_{\mathbf D}, \widehat{P}_{\mathbf X_\param})\rightarrow D_\mathrm{JS}(\widehat{P}_{\mathbf D}, P_\param)$ when the simulation count $\nsms\rightarrow\infty$ \citep{corremkos2}, but otherwise the discrepancies calculated based on individual simulation results $D_\mathrm{JS}(\widehat{P}_{\mathbf D}, \widehat{P}_{\mathbf X_\param})$ are best understood as $E[D_\mathrm{JS}(\widehat{P}_{\mathbf D}, \widehat{P}_{\mathbf X_\param})]$ observed with additive noise.
This makes optimization difficult.
Since running simulations can be expensive, we want to limit the total simulation count, and we need an optimization method that can balance between running simulations with the same parameter values to improve the expected JSD estimates locally and running simulations with new parameter values to find the expected JSD minimum globally. 
The experiments carried out in this work use Bayesian optimization.
For a tutorial and review, see \citet{frazier2018tutorial}.

\section{   Likelihood-free Model Choice  based on  JSD   Razor    }\label{jsdcritsugg}\label{bicschwartz}

As outlined in the Introduction, Section \ref{ppo:sec:intro}, the  purpose of statistical model selection is to select from 
a set of alternative explanations  or models, the one that best explains the data, here categorical   $\mathbf{D}$. 
The task is not elementary, as there are two conflicting  requirements  of a  good model,  namely  that of  generalizability and 
that of goodness of fit.  Here goodness of fit measures how well a model fits the observed  $\mathbf{D}$.   By generalizability  we refer to the capability of the model to fit   well  novel data sets. 

Occam$^{,}$s Razor  is  known as the  dictum  that   the simpler model, for example a model with fewer parameters, is to be preferred  (a.k.a. the principle of parsimony).  In terms of information theory the Occam$^{,}$ Razor  proposes   that the shortest description of the data is the best model. 

In this section we   start  from the interpretation   of Occam$^{,}$s  Razor 
in  \citet{balasubramanian1996geometric}.   JSD is related to   this Razor by bounds  for the model evidence.  
We derive an expression for the JSD Razor by finding the total differential and Hessian of  JSD w.r.t.  the parameters. Then  the multivariable  Laplace approximation is applied to get an approximate expression of the JSD Razor    to be minimized by choice of  model family.

\subsection{ Bounds  for the Model Evidence}\label{infbd}

For  a finite number $L$ of alternative  models  $\mathbb{M}_{p}^{(l)}$ for $\mathbf{D}$,  the Bayesian rule  of  selecting  among  them is  to 
pick  the model that maximizes  the posterior probability given   $\mathbf{D}$. If the  models have equal prior probabilities,  
this means selecting the model that maximizes the the model evidence as defined in Equation~(\ref{datallikelihood}).   

Next we drop for convenience of writing the superscript in $\mathbb{M}_{p}^{(l)}$.  
It has been shown in \citet{corremkos2} that  $\widehat{P}_{\mathbf{D}}$ is  Bayes sufficient summary statistic  
of the data  $\mathbf{D}$.  This follows  effectively  by  the multinomial probability 
\begin {equation}\label{hoffdmultid}
 P_{\theta} (\widehat{P}_{\mathbf{D}}) =  \frac{n_{o}!}{  \prod_{j=1}^{k}n_{o,j}!} 
P_{\theta}\left(\mathbf{D} \right).
\end{equation}
Thus we define     the model evidence based on the  sufficient summary of   $\mathbf{D}$, $P\left(\widehat{P}_{\mathbf{D}}\mid \mathbb{M}_{p}\right )$,  by   modification of   Equation~(\ref{datallikelihood}) as   
$$
P\left(\widehat{P}_{\mathbf{D}}\mid \mathbb{M}_{p}\right ):= \frac{n_{o}!}{\prod_{j=1}^{k} n_{o,j} !} 
P\left(\mathbf{D}\mid \mathbb{M}_{p}\right ) =  \int_{\Theta} P_{\theta} (\widehat{P}_{\mathbf{D}}) p(\theta)d\theta ,
$$
cf., Equation~(\ref{sllnh}).
  In view of Equation~(\ref{hoffdmultid}) and  Equation~(\ref{typen111})  and evaluation of $e^{ -n_{o}H\left(\widehat{P}_{\mathbf{D}}\right)}$ by definition of the Shannon entropy  we obtain 
\begin{equation}
 \int_{\Theta} P_{\theta} (\widehat{P}_{\mathbf{D}}) p(\theta)d\theta  = \frac{n_{o}!}{  \prod_{j=1}^{k}n_{o,j}!} \prod_{j=1}^{k}\left(\frac{n_{o,j}}{n_{o}}\right)^{n_{o,j}} \int_{\Theta} e^{-n_{o} D_{\rm KL} \left(  \widehat{P}_{\mathbf{D}},  P_{\theta}\right)} p(\theta) d\theta. \nonumber 
\end{equation}
The multiplicative factor  in front of the integral 
in the right hand side is $<1$ since it  is the   probability of  the event  $\underline{ \xi}=\left (n_{1},\ldots, n_{k} \right)$ w.r.t. to the multinomial distribution with parameters $ n_{o}$  
and  $\left(\frac{n_{1}}{ n_{o}},  \ldots, \frac{n_{k}}{n_{o}}\right)$.  By   Lemma  \ref{lemmapinskerbd}  we have found 
\begin{equation}\label{eq:model_evidence_and_razor_upper_bound}
 \int_{\Theta} P_{\theta} (\widehat{P}_{\mathbf{D}}) p(\theta)d\theta  \leq  \int_{\Theta} e^{-n_{o} D_{\rm KL} \left(  \widehat{P}_{\mathbf{D}},  P_{\theta}\right)} p(\theta) d\theta \leq  \int_{\Theta} e^{-2n_{o}  D_{\rm JS}\left(  \widehat{P}_{\mathbf{D}},  P_{\theta}\right)}p(\theta) d\theta.
\end{equation}
The idea is to use  the integral  in the right  hand side of the second inequality   to find  an implementable  criterion for likelihood-free simulator-based 
models choice.  For a special prior density  $p(\theta)$  the integral  will below be called the JSD-Razor. 
\subsection{Definition of the JSD Razor and Outline}
Let $I(\theta)$  be  the Fisher  information matrix (see Section \ref{antaganden}).   Consider the prior 
$$
p(\theta) = \frac{\sqrt{\det I(\theta)}}{ V\left( \Theta \right)},
$$
where  $ V\left( \Theta \right):=\int_{\Theta}  \sqrt{\det I(\theta)} d \theta$ is assumed to exist.  Hence $p(\theta)$ is 
Jeffreys$^{,}$ prior, which is rigorously constructed in   \citet{balasubramanian1996geometric}  by a convergence argument  from a discrete uniform prior  on 
a finite number of $  P_{\theta}$   indistinguishable (in a  sense made precise in \citet{balasubramanian1996geometric})  from   $P_{\theta_{o}}$.
Then we define the JSD-Razor ${\rm R}_{n_{o}}\left(\mathbb{M}_{p} \right)$  by
\begin{equation}\label{jsdrazordef} 
{\rm R}_{n_{o}}\left(\mathbb{M}_{p}\right):= 
\int_{\Theta} e^{-2n_{o}  D_{\rm JS}\left(  \widehat{P}_{\mathbf{D}}, P_{\theta}\right) } \frac{ \sqrt{\det I(\theta)}}{V\left( \Theta \right)} d\theta.
\end{equation}
Occam$^{,}$s   Razor as introduced  in   \citet[Equation (35), p.~20]{balasubramanian1996geometric}  becomes
  $$ \int_{\Theta} e^{-n_{o} D_{\rm KL} \left(  \widehat{P}_{\mathbf{D}},  P_{\theta}\right)} \frac{ \sqrt{\det I(\theta)}}{V\left( \Theta \right)} d\theta.$$ 
We have seen above (Equation~\ref{eq:model_evidence_and_razor_upper_bound}) that Occam$^{,}$s   Razor  is bounded upwards by ${\rm R}_{n_{o}}\left(\mathbb{M}_{p} \right)$.

In the following  Section \ref{jsdrazor221} we argue for   and derive  from  maximization of  the JSD-Razor ${\rm R}_{n_{o}}\left(\mathbb{M}_{p}^{(l)} \right)$   two rules for model choice. The criterion  ${\rm SIC}_{\rm JSD} \left( \mathbb{M}_{p}^{(l)}\right )$ below  is the cruder version of  the asymptotics of the JSD-Razor ${\rm R}_{n_{o}}\left(\mathbb{M}_{p} \right)$.  It chooses  the model $\mathbb{M}_{p}^{(l)}$ that minimizes 
\begin{equation}\label{sicjsdo}
{\rm SIC}_{\rm JSD} \left( \mathbb{M}_{p}^{(l)}\right ) :=2 n_{o} D_{\rm JS} \left( \widehat{P}_{ \mathbf{D}}, P_{\widehat{\theta}^{(l)}_{\rm JSD}}\right) + {\rm dim} (\Theta^{(l)}) \ln \sqrt{\frac{n_{o}}{8\pi}}
\end{equation}
among a  finite number $L$ of alternative  models  $\mathbb{M}_{p}^{(l)}$ for $\mathbf{D}$,  where it is understood that   $n_{o} >8 \pi $, and where we define  the minimum JSD estimate of  $\theta \in \Theta^{(l)}$  by 
\begin{equation}\label{minjsd}
 \widehat{\theta}^{(l)}_{\rm JSD} = \widehat{\theta}^{(l)}_{\rm JSD} (\mathbf{D})=  {\rm argmin}_{\theta \in \Theta^{(l)}} D_{\rm JS} \left(  \widehat{P}_{\mathbf{D}},  P_{\theta}\right).
\end{equation}
It is shown in \citet{corremkos2} that $ \widehat{\theta}^{(l)}_{\rm JSD}$ exists for compact  $\Theta^{(l)}$. Asymptotic (in $n_{o}$) properties of   $ \widehat{\theta}^{(l)}_{\rm JSD}$ can be extracted from  the results on general  minimum $\phi$-divergence estimates  in \citet{morales1995asymptotic} for prescribed models. The criterion in Equation~(\ref{sicjsdo})  contains, as it should,   the trade-off between fit and model dimension:  the JSD-fit  will be smaller in  a model with a larger parameter space.  
In  Equation~(\ref{sicjsdo}) we have a   criterion which  is  computable for choice between simulator-based  models   in  the sense that   the computational minimization of $E_{P_\theta}\left[D_{\rm JSD}(\widehat{P}_{\mathbf {D}},\widehat{P}_{\mathbf {X}_{\theta}})\right]$  by the software function BOLFI should approximately  find  $\widehat{\theta}_{\rm JSD}$, as  defined by Equation~(\ref{minjsd}). 
 
\subsection{Step One  for Derivation of  ${\rm SIC}_{\rm JSD}$ : Total Differential and Hessian of JSD w.r.t  $\theta$} 
The following result is well-known, see, e.g.,   \citet[p.~355]{morales1995asymptotic}.
 \begin{lemma}\label{infemma}
 \begin{equation}\label{nablanotexp}
I(\theta)=A(\theta)^{T}A(\theta).
\end{equation}
Under Assumption \ref{invfisher},  the  matrix  $I(\theta)$ is invertible. 
\end{lemma}
It is  appropriate  here to use the simplex  map  from Equation~(\ref{trmap})  for the distributions in $\mathbb{P}$.  In  other words, we  shall work with 
functions of vectors in $\triangle_{k-1}$.  Then 
 ${\bf p}_{\theta} =\left(p_{1}(\theta), \ldots, p_{k}(\theta) \right)= \triangle\left(P_{\theta} \right)$, $\widehat{\mathbf{p}} =\triangle \left( \widehat{P}_{\mathbf{D}}\right)$, and  we obtain by means of Iverson bracket (Equation~\ref{iversonbr}) and  Equation~(\ref{jsddivdef})  that 
\begin{equation}\label{jsfphimapo}
D_{\rm JS}\left(  \widehat{\mathbf{p}}, {\bf p}_{\theta}\right)
=D_{\rm JS}\left(  \widehat{P}_{\mathbf{D}} , P_{\theta}\right).
\end{equation}
Since we are going to  at first deal with   this quantity   by partial derivatives  w.r.t.  $\theta$  for fixed $\widehat{\mathbf{p}}$, we introduce for ease of writing the function 
\begin{equation}\label{jsfphimap}
\Phi\left( \widehat{\mathbf{p}} ,  \theta   \right): =D_{\rm JS}\left(  \widehat{\mathbf{p}}, {\bf p}_{\theta}\right)
\end{equation}
defined on   ${\triangle}_{k-1}\times \Theta $.  We compute with  dropping  the hat  so that  $\mathbf{p} \leftarrow \widehat{\mathbf{p}}$.
We recall the generating $\phi$  of symmetric JSD in Equation~(\ref{jsddiv}) 
$$
\phi(u)= \frac{1}{2} u \ln u - \frac{1}{2}\left( u +1 \right)\ln  \left(   \frac{1}{2} u +  \frac{1}{2} \right),  0 \leq u < +\infty.
$$ 
The first derivative is 
\begin{equation}\label{jsddivder}
\phi^{'}(u)=   \frac{1}{2 }\ln u -\frac{1}{2 }\ln \left(\frac{u}{2 } + \frac{1}{2 } \right), 0  <  u < +\infty, 
\end{equation}
We can rewrite 
\begin{equation}\label{jsddivder3}
\phi^{'}(u)=   \frac{1}{2 }\ln \left( \frac{u}{\left(\frac{u}{2 } + \frac{1}{2 } \right)}\right), 
\end{equation} 
and this gives 
\begin{equation}\label{jsddivder4}
\phi(u)=   u\phi^{'}(u) -\frac{1}{2 }\ln \left(\frac{u}{2 } + \frac{1}{2 } \right).
\end{equation} 
Next we  get 
\begin{equation}\label{jsddiv2}
  \phi^{''}(u)= \frac{1}{2 u(u+1)},  0 <  u < +\infty.
\end{equation}
\begin{lemma} \label{birchloglike}
Assume 
 ${\bf p} (\theta) \in \stackrel{o}{\triangle}_{k-1} $.
Let us define   the  $ 1 \times k$ vector 
\begin{equation}\label{graddual}
\Phi_{\phi}\left(  {\bf p},  \theta   \right) =\left(  \phi^{'}\left( \frac{p_{1}(\theta)}  {p_{1}}  \right), \ldots, \phi^{'}\left( \frac{p_{k}(\theta)} {p_{k}}   \right)\right).
\end{equation} 
Then  the $1 \times d$ total differential  $\frac{\partial}{\partial \theta }\Phi\left(  {\bf p},  \theta   \right)$  is given by 
\begin{equation}\label{gradth}
\frac{\partial}{\partial \theta }\Phi\left(  {\bf p},  \theta   \right)= \Phi_{\phi}\left(  {\bf p},  \theta   \right) J\left(\theta \right),
\end{equation} 
where $J\left(\theta \right)$  is the Jacobian  in Equation~(\ref{nablanot3}). 
 \end{lemma} 
The  straightforward computational proof is found in Appendix  \ref{diffintvajdapist}.

Let us define the $k \times k$ diagonal matrix  
\begin{equation}\label{nabladiag1}
\Lambda({\bf p}, \theta):= {\rm diag}\left(   \frac{1}{\sqrt{(p_{1}+p_{1}(\theta))}},\ldots,   \frac{1}{\sqrt{(p_{k}+p_{k}(\theta))}}  \right), 
\end{equation}
And let us define with the Jacobian $J(\theta)$  in  Equation~(\ref{nablanot3})
\begin{equation}\label{ajsdfaktor}
A({\bf p}, \theta) :=  \Lambda\left({\bf p}, \theta \right)J(\theta).
\end{equation}

\begin{lemma}\label{jsdhessian} Assume ${\bf p}  \in \stackrel{o}{\triangle}_{k-1} $  and   ${\bf p} (\theta) \in \stackrel{o}{\triangle}_{k-1} $.
The Hessian matrix  of  $\Phi \left(  {\bf p},  \theta  \right)$ is the  $ d\times d$ matrix  $H_{\Phi}\left(  {\bf p},  \theta  \right)$  with elements 
given by 
\begin{equation}\label{jsdehesselem}
\frac{\partial^{2}}{\partial \theta_{l}\partial \theta_{j}}\Phi\left(  {\bf p},  \theta \right)= \sum_{i=1}^{k}   \phi^{'}\left( \frac{p_{i}(\theta)}  {p_{i}}  \right) \frac{\partial^{2}}{\partial \theta_{l} \partial \theta_{j}}p_{i}(\theta) + \frac{1}{2} I_{lj}(\theta) -\frac{1}{2} \left[A({\bf p}, \theta)^{T}A({\bf p}, \theta)\right]_{lj}. 
\end{equation} 
\end{lemma}
The proof is found in the Appendix \ref{diffintvajdapist}.  We introduce the $d \times d$ matrix 
$$
H_{p_{1}, \ldots, p_{k}} \left( \theta   \right) :=  \left[\sum_{i=1}^{k}   \phi^{'}\left( \frac{p_{i}(\theta)}  {p_{i}}  \right) \frac{\partial^{2}}{\partial \theta_{l} \partial \theta_{j}}p_{i}(\theta) \right]_{l=1,j=1}^{d,d}.
$$
Then Equation~(\ref{jsdehesselem})  becomes     the matrix  equality.
\begin{proposition}\label{hessmatr} Assume ${\bf p}  \in \stackrel{o}{\triangle}_{k-1} $  and   ${\bf p} (\theta) \in \stackrel{o}{\triangle}_{k-1} $. The Hessian matrix w.r.t.  $\theta$  of 
$\Phi\left( \bf{p} ,  \theta   \right)$  in Equation~(\ref{jsfphimap})
 \begin{equation}\label{jsdehesselem2}
  H_{\Phi}\left(  {\bf p},  \theta  \right) = H_{p_{1}, \ldots, p_{k}} \left( \theta   \right)  +  \frac{1}{2} I(\theta) -\frac{1}{2} A({\bf p}, \theta)^{T}A({\bf p}, \theta).
\end{equation} 
\end{proposition}

\subsection{Step Two   for Derivation of  ${\rm SIC}_{\rm JSD}$: Laplace Approximation of the JSD-Razor}\label{laplarazor}
We  adapt  for the  current  setting some pertinent  standard  results of  \citet[Ch. 5]{breitung2006asymptotic} or  \citet[Ch. XI]{wong2001asymptotic} for Laplace approximation of multivariate integrals, in this case for 
 $  {\rm R}_{n_{o}}\left(\mathbb{M}_{p} \right) $ in  Equation~(\ref{jsdrazordef}). The work in    \citet{lapinski2019multivariate}  adds  convergence rates of this  approximation. 
 We  use the map  
$\Phi\left( \widehat{\mathbf{p}} ,  \theta   \right)$  in Equation~(\ref{jsfphimap}). Hence by  Equation~(\ref{minjsd})  we write  $ \widehat{\theta}_{\rm JSD} = \widehat{\theta}_{\rm JSD} (\mathbf{D})=$ $  {\rm argmin}_{\theta \in \Theta}\Phi\left( \widehat{\mathbf{p}} ,  \theta   \right)$.

\begin{lemma}\label{laplcelemma} Assume  Assumptions \ref{kontdiff}--\ref{invkontbirch}. Assume that  $\Theta$ is a compact subset of  $\mathbb{R}^{d}$ and that  $ \Phi\left( \widehat{\mathbf{p}} ,  \theta   \right)$  has a unique minimum at $\widehat{\theta}_{\rm JSD}$   in the interior of  $\Theta$
and that  the Hessian   $H_{\Phi}\left(  \widehat{\bf p},  \widehat{\theta}_{\rm JSD}  \right)$ is positive definite, and that the prior density $p(\theta)$ is a continuous  function of $\theta$.  Then it holds that
\begin{equation}\label{sicjsd2}
 {\rm R}_{n_{o}}\left(\mathbb{M}_{p} \right)  \approx \left(\frac{2\pi}{n_{o}}\right)^{d/2}e^{-2n_{o}  D_{\rm JS}\left(  \widehat{P}_{\mathbf{D}},  P_{ \widehat{\theta}_{\rm JSD}}\right) }
 \frac{\sqrt{\det I \left(\widehat{\theta}_{{\rm JSD} }\right)}}
{\sqrt{   \det H_{\Phi}\left(  \widehat{\bf p}, \widehat{\theta}_{\rm JSD} \right)}} \frac{1}{V\left( \Theta \right)}.
\end{equation}
\end{lemma}
\begin{proof}
Due to the  differentiability  properties  of   $\phi_{\rm JS}(u)$ in Equation~(\ref{jsddiv}) and the assumptions   the map    $\Phi\left( \widehat{\mathbf{p}} ,  \theta   \right) $ is a twice differentiable as a function of $\theta$ in the interior of $\Theta$, as has been checked in the preceding section. 
Then the  equality  Equation~(\ref{sicjsd2}) is valid for $ {\rm R}_{n_{o}}\left(\mathbb{M}_{p} \right)$ in  Equation~(\ref{jsdrazordef}), as follows   by \citet[Thm 4.1, p.~56]{breitung2006asymptotic}  or   \citet[Thm 3. pp.~494$-$495]{wong2001asymptotic}.
\end{proof}
\begin{lemma}\label{hessenapprsats}
Let the assumptions of Lemma  \ref{laplcelemma}  hold for  any $n_{o}$.   Assume that there is $ P_{\theta_{o}} \in \mathbb{M}$  such that   $\mathbf{D} \sim  P_{\theta_{o}} $. Then  with  $ P_{\theta_{o}}$-probability   one, 
\begin{equation}\label{hessenappr}
H_{\Phi}\left(  \widehat{\bf p}, \widehat{\theta}_{\rm JSD} \right) \approx  \frac{1}{4} I\left(\widehat{\theta}_{\rm JSD} \right).
\end{equation}
for all  large  $n_{o}$, where  $I\left(\widehat{\theta}_{\rm JSD} \right)$ is the Fisher  information matrix defined in
Equation~(\ref{ijfisher}) evaluated at $\widehat{\theta}_{\rm JSD}$. 
\end{lemma}

\begin{proof} We know by Lemma \ref{birchbd} that 
\begin{equation}\label{jsdl22}
\frac{\sqrt{2}}{4}  || \widehat{P}_{\mathbf{D}} - P_{\widehat{\theta}_{\rm JSD}}||_{2} \leq   D^{1/2}_{\rm JS}\left(  \widehat{P}_{\mathbf{D}}, P_{\widehat{\theta}_{\rm JSD}}\right).
\end{equation}
In force of Proposition \ref{sqrmetric}, we can use the triangle inequality to the effect that  
\begin{equation}\label{jsdtrian}
  D^{1/2}_{\rm JS}\left(  \widehat{P}_{\mathbf{D}}, P_{\widehat{\theta}_{\rm JSD}}\right) \leq    D^{1/2}_{\rm JS}\left(  \widehat{P}_{\mathbf{D}}, P_{\theta_{o}}\right)   + D^{1/2}_{\rm JS}\left( P_{\theta_{o}} ,P_{\widehat{\theta}_{\rm JSD}}\right). 
\end{equation}   
It has been shown in \citet[Lemma 5.1, Proposition 5.2]{corremkos} that  $D^{1/2}_{\rm JS}\left( P_{\theta_{o}} ,P_{\widehat{\theta}_{\rm JSD}}\right) \rightarrow 0$,  $ P_{\theta_{o}}$  a.s., as $n_{o} \rightarrow +\infty$, 
when  $ P_{\theta_{o}} \in \mathbb{M}$ and  $\mathbf{D} \sim  P_{\theta_{o}} $. Let us recall the total variation distance $V(P,Q) $  defined in  (\ref{totvar}).  It holds that  
$D^{1/2}_{\rm JS}\left(  \widehat{P}_{\mathbf{D}}, P_{\theta_{o}}\right) \rightarrow 0$, if  $V\left(  \widehat{P}_{\mathbf{D}}, P_{\theta_{o}}\right)$ $\rightarrow 0$, 
  $ P_{\theta_{o}}$  a.s., as $n_{o} \rightarrow +\infty$, which convergence is well known,  for a proof see 
e.g., Appendix  A of  \citet{corremkos2}, as    $ P_{\theta_{o}} \in \mathbb{M}$ and  $\mathbf{D} \sim  P_{\theta_{o}} $.
Hence   in Equation~(\ref{jsdl22}), the norm 
\begin{equation}
  || \widehat{P}_{\mathbf{D}} - P_{\widehat{\theta}_{\rm JSD}}||_{2} = \sqrt{\sum_{i=1}^{k}  \left(  \widehat{p}_{i} -  p_{i}\left(\widehat{\theta}_{\rm JSD}   \right) \right)^{2}}  \rightarrow 0,
\end{equation}
 $ P_{\theta_{o}}$  a.s., as $n_{o} \rightarrow +\infty$.

The Hessian  $H_{\Phi}\left(  \widehat{\bf p}, \widehat{\theta}_{\rm JSD} \right)$  has by  Proposition 
\ref{hessmatr}, Equation~(\ref{jsdehesselem2}),  by Equation~(\ref{hessele})  and Equation~(\ref{fisinff})  the elements
\begin{eqnarray}\label{hessele1}
\frac{\partial^{2}}{\partial \theta_{l}\partial \theta_{j}}\Phi\left(  \widehat{\bf p},  \widehat{\theta}_{\rm JSD}  \right)&=&\sum_{i=1}^{k}   \phi^{'}\left( \frac{p_{i}\left(\widehat{\theta}_{\rm JSD} \right)}  {\widehat{p}_{i}}  \right) \frac{\partial^{2}}{\partial \theta_{l} \partial \theta_{j}}p_{i}\left(\widehat{\theta}_{\rm JSD} \right) \nonumber \\
&  &  \\
& &+ \frac{1}{2} I_{lj}\left(\widehat{\theta}_{\rm JSD} \right)   -  \frac{1}{2}\sum_{i=1}^{k} 
\frac{ \frac{\partial}{\partial \theta_{l}}p_{i}\left(\widehat{\theta}_{\rm JSD} \right) \frac{\partial}{\partial \theta_{j}}p_{i}\left(\widehat{\theta}_{\rm JSD} \right)   }{\left(\widehat{p}_{i}+p_{i}\left(\widehat{\theta}_{\rm JSD} \right) \right)}.  \nonumber
\end{eqnarray}
Since   $\left(  \widehat{p}_{i} -  p_{i}\left(\widehat{\theta}_{\rm JSD}   \right) \right)^{2}\rightarrow 0 $,  
$ \phi^{'}\left( \frac{p_{i}\left(\widehat{\theta}_{\rm JSD} \right)}  {\widehat{p}_{i}}  \right) \approx $
$ \phi^{'}\left( 1\right) =0$, where we used Equation~(\ref{jsddivder3}).
For the same reason  
$$
\frac{ \frac{\partial}{\partial \theta_{l}}p_{i}\left(\widehat{\theta}_{\rm JSD} \right) \frac{\partial}{\partial \theta_{j}}p_{i}\left(\widehat{\theta}_{\rm JSD} \right)   }{\left(\widehat{p}_{i}+p_{i}\left(\widehat{\theta}_{\rm JSD} \right) \right)}  \approx \frac{1}{2}
\frac{ \frac{\partial}{\partial \theta_{l}}p_{i}\left(\widehat{\theta}_{\rm JSD} \right) \frac{\partial}{\partial \theta_{j}}p_{i}\left(\widehat{\theta}_{\rm JSD} \right)   }{p_{i}\left(\widehat{\theta}_{\rm JSD} \right)}.
$$
Now Lemma  \ref{infemma}, or  the expression   Equation~(\ref{ijfisher3}) in its proof, gives 
$$
 \frac{1}{2}\sum_{i=1}^{k} 
\frac{ \frac{\partial}{\partial \theta_{l}}p_{i}\left(\widehat{\theta}_{\rm JSD} \right) \frac{\partial}{\partial \theta_{j}}p_{i}\left(\widehat{\theta}_{\rm JSD} \right)   }{\left(\widehat{p}_{i}+p_{i}\left(\widehat{\theta}_{\rm JSD} \right) \right)} \approx \frac{1}{4}I_{lj}\left(\widehat{\theta}_{\rm JSD} \right). 
$$
When  the approximate expressions  above  have been  applied  in Equation~(\ref{hessele1}), the right hand side of Equation~(\ref{hessenappr})
is obtained. \end{proof}  

\subsection{ Step Three    for Derivation of  ${\rm SIC}_{\rm JSD}$: Two Versions  }\label{jsdrazor221}
Next we produce the rule  in Equation~(\ref{sicjsdo}).
 We want to find  the model   $\mathbb{M}_{p}^{(l)}$, which minimizes   $- \ln {\rm R}_{n_{o}}\left(\mathbb{M}_{p}^{(l)} \right)$. Following    \citet[(35), p.~20]{balasubramanian1996geometric} we  define 
\begin{equation}\label{cvol}
V_{c}(\Theta):=  \left(\frac{2\pi}{n_{o}}\right)^{d/2} \frac{\sqrt{\det I \left( \widehat{\theta}_{{\rm JSD} }\right)}}{\sqrt{   \det H_{\Phi}\left(  \widehat{\bf p}, \widehat{\theta}_{\rm JSD} \right)}}.
\end{equation}
Here  $\ln \frac{V(\Theta)}{V_{c}(\Theta)}$ acts as   a penalty 
for model complexity in the  geometric sense  of  model volume, see \citet{balasubramanian1996geometric} and   \citet{balasubramanian2005mdl} and  \citet{myung2000counting}.

 By  Lemma \ref{laplcelemma}, 
\begin{equation}\label{jsdrazor22}
 - \ln {\rm R}_{n_{o}}\left(\mathbb{M}_{p}\right) \approx 2n_{o}  D_{\rm JS}\left(  \widehat{P}_{\mathbf{D}}, P_{ \widehat{\theta}_{\rm JSD}}\right)  +\ln \frac{V(\Theta)}{V_{c}(\Theta)}, 
\end{equation}
By an expansion  we get  
$$
\ln \frac{V(\Theta)}{V_{c}(\Theta)}= \frac{d}{2} \ln \frac{n_{o}}{2\pi} + \ln \int_{\Theta}  \sqrt{\det I(\theta)} d \theta + \frac{1}{2} \ln \frac{ \det H_{\Phi}\left(  \widehat{\bf p}, \widehat{\theta}_{\rm JSD} \right)} {\det I \left(\widehat{\theta}_{{\rm JSD} }\right)}.
$$
By Lemma \ref{hessenapprsats}, Equation~(\ref{hessenappr}),  and the rules for determinants we have 
$$
 \det H_{\Phi}\left(  \widehat{\bf p}, \widehat{\theta}_{\rm JSD} \right) \approx  \frac{1}{4^{d}} \det  I\left(\widehat{\theta}_{\rm JSD} \right), 
$$ 
and  
\begin{equation}\label{notcoarse}
\ln \frac{V(\Theta)}{V_{c}(\Theta)} \approx \frac{d}{2}  \ln \frac{n_{o}}{2\pi} + \ln \int_{\Theta}  \sqrt{\det I(\theta)} d \theta 
- d \ln2, 
\end{equation}
whereby Equation~(\ref{sicjsdo}) is obtained by dropping $ \ln \int_{\Theta}  \sqrt{\det I(\theta)} d \theta$.  An example of  explicit computation of  
$\det I(\theta)$ is presented in Appendix \ref{birchmapmain} (Example \ref{multicateg3}).
\section{Properties of ${\rm SIC}_{\rm JSD}$} \label{propjsd}
In this section we  study the  properties of  ${\rm SIC}_{\rm JSD} \left( \mathbb{M}_{p}^{(l)}\right )$ as approximately obtained from JSD Razor  in the preceding section  and stated  in Equation~(\ref{sicjsdo}).  We make a comparison of  ${\rm SIC}_{\rm JSD}$  with the well known Schwarz$^{,}$s information criterion (SIC) for model  determination. This criterion is also known as Bayesian information criterion (BIC).  Then we prove the  consistency 
of   ${\rm SIC}_{\rm JSD} \left( \mathbb{M}_{p}^{(l)}\right )$ for nested models  both when the true data source has a model  included in one of the models and when it is  not. 
 
\subsection{An Upper Bound by  Two-Part MDL  }
 Let $\widehat{\theta}^{(l)}_{\rm ML}$   denote the maximum likelihood estimate of $\theta$.
The minimization of 
\begin{equation}\label{scicrit}
{\rm SIC}\left( \mathbb{M}_{p}^{(l)}\right ) :=- \ln P_{\widehat{\theta}^{(l)}_{\rm ML}}\left( \mathbf{D} \right)  + \frac{{\rm dim} (\Theta^{(l)})}{2}\ln n_{o}
\end{equation}
as a function on the set of models $\mathbb{M}_{p}^{(l)}$, $l=1, \ldots, L$,  is known as Schwarz$^{,}$s information criterion  for model  determination,  see   \citet{cavanaugh1999generalizing} for  the derivation,  \citet{neath2012bayesian}  for a recent survey of applications, and  \citet[pp.~5--6]{rissanen2007information} and  \citet[section 7.2.3, pp.~352--353]{robert2007bayesian}  for critical remarks.  It is needless to point out that minimization of  ${\rm SIC}\left( \mathbb{M}_{p}\right )^{(l)} $ is not available for implicit models and likelihood-free inference. 

 Rissanen   proved    that the   model achieving the minimum of  ${\rm SIC}\left( \mathbb{M}_{p}^{(l)}\right ) $ gives the least redundant  coding  possible of $\mathbf{D}$   amongst all universal codes, where    optimal quantization of  $\Theta$  is achieved by using accuracy of order $1/\sqrt{n_{o}}$, see \citet[p.~2]{roos2017minimum} for the result and further references. 

By    Section \ref{topsoecode} above,  
 the  term $2 n_{o} D_{\rm JS} \left( \widehat{P}_{ \mathbf{D}}, P_{\widehat{\theta}^{(l)}_{\rm JSD}}\right) $ 
in Equation~(\ref{sicjsdo}) can be regarded as   redundance  in a different sense.  In other words,  ${\rm SIC}_{\rm JSD} \left( \mathbb{M}_{p}^{(l)}\right )$  is a sum of  redundance and a penalty term of basically same  form  as in  
the two part redundance code length in Equation~(\ref{scicrit}).  The next proposition  is suggested by  Section  \ref{infbd}.

\begin{proposition}
With  SIC as  defined in  Equation~(\ref{scicrit})  and ${\rm SIC}_{\rm JSD}$  from Equation~(\ref{sicjsdo}) it holds that  
\begin{equation}\label{scicrit2}
{\rm SIC}_{\rm JSD} \left( \mathbb{M}_{p}^{(l)}\right ) < {\rm SIC}\left( \mathbb{M}_{p}^{(l)}\right ).
\end{equation}
\end{proposition}
\begin{proof} By  definition of $\widehat{\theta}_{\rm JSD}$  in Equation~(\ref{minjsd}) we have  for $\widehat{\theta}_{\rm ML}$, the maximum likelihood
estimate of   $\theta$ based on  $\mathbf{D}$,   
$$
2n_{o}D_{\rm JS} \left( \widehat{P}_{ \mathbf{D}}, P_{\widehat{\theta}_{\rm JSD}}\right)  \leq   
2n_{o}D_{\rm JS} \left( \widehat{P}_{ \mathbf{D}}, P_{\widehat{\theta}_{\rm ML}}\right)
$$
and by Lemma  \ref{lemmapinskerbd} and definition of  $D_{\rm KL} $ (Equation~\ref{klinformation})
\begin{eqnarray}
2n_{o}D_{\rm JS} \left( \widehat{P}_{ \mathbf{D}}, P_{\widehat{\theta}_{\rm ML}}\right) &  \leq & 
n_{o}D_{\rm KL} \left( \widehat{P}_{ \mathbf{D}}, P_{\widehat{\theta}_{\rm ML}}\right)  \nonumber \\
& =& \sum_{j=1}^{k} n_{o,j} \ln  \left( \frac{n_{o,j}}{n_{o}} \right) -   \sum_{j=1}^{k} n_{o,j} \ln  \left( p_{j} \left(  \widehat{\theta}_{\rm ML}   \right) \right)  \nonumber \\
& <  & - \sum_{j=1}^{k} n_{o,j} \ln  \left( p_{j} \left(  \widehat{\theta}_{\rm ML}  \right) \right), \nonumber 
\end{eqnarray}
since $\sum_{j=1}^{k} n_{o,j} \ln  \left( \frac{n_{o,j}}{n_{o}} \right) < 0$.  Hence we have established  that 
$$
2n_{o}D_{\rm JS} \left( \widehat{P}_{ \mathbf{D}}, P_{\widehat{\theta}_{\rm JSD}}\right)  <  -\ln P_{  \widehat{\theta}_{\rm ML}} \left(\mathbf{D} \right).
$$
 Since 
$ \frac{{\rm dim} (\Theta^{(l)})}{2}) \ln \frac{n_{o}}{8\pi}    < $ $ \frac{{\rm dim} (\Theta^{(l)})}{2} \ln n_{o}$, the inequality  in the proposition  holds by the  definitions  in Equations~(\ref{sicjsdo})  and  (\ref{scicrit}).   \end{proof}
The criteria in  Equations~(\ref{sicjsdo})  and  (\ref{scicrit}) are also inherently connected  via the already cited  fact that   $\widehat{\theta}_{\rm JSD}$ and  $\widehat{\theta}_{\rm ML}$   are asymptotically equal, when  $n_{o} \rightarrow +\infty$, as  shown in  \citet{corremkos}. 

The inequality  Equation~(\ref{scicrit2}) tells that  if  for  the optimal model $\mathbb{M}_{p}^{(l^{\ast})}$ w.r.t  ${\rm SIC}$ the minimum  value of  ${\rm SIC}\left( \mathbb{M}_{p}^{(l^{\ast})}\right )$  is very small, then  $\mathbb{M}_{p}^{(l^{\ast})}$ is likely to be the minimizer of  ${\rm SIC}_{\rm JSD} $, too. 

\subsection{Consistency of the  JSD-Razor Rule    }\label{sciconsisten}
We show next a consistency property of the JSD-Razor  model selection  criterion for nested models. Consistency means that the criterion will
asymptotically select, with probability one, amongst candidate models  $\mathbb{M}_{p}^{(l)}$, $l=1,  \ldots, L$,  the   the most parsimonious   $\Theta^{(l)}$  model 
containing the true generating distribution. 
From a theoretical point of view such  consistency is  a  very strong optimality property of the JSD-Razor model choice. 
\begin{proposition}\label{cons1}
The  models $\mathbb{M}_{p}^{(l)}$, $l=1,  \ldots, L$ are  nested    
$$
\Theta^{(1)} \subseteq \Theta^{(2)} \subseteq  \ldots \subseteq  \Theta^{(L)}
$$
and $ {\rm dim}\left (\Theta^{(1)}  \right))\leq  {\rm dim} \left(\Theta^{(2)}\right) \leq \ldots  \leq {\rm dim} \left(\Theta^{(L)} \right)( <k)$.
Suppose that $l_{o}$ is the smallest integer in $\{1, \ldots, L\}$ such that  the true probability $ P_{\theta_{o}} \in \mathbb{M}^{\left(l_{o}\right)}$.
Then  for  large $n_{o}$  
\begin{equation}\label{jsdkonsist}
{\rm SIC}_{\rm JSD} \left( \mathbb{M}_{p}^{\left(l_{o}\right)}\right ) \leq {\rm SIC}_{\rm JSD} \left( \mathbb{M}_{p}^{\left(l\right)}\right ) \quad \text{for every  $l \neq  l_{o}$ }
\end{equation}
with   $ P_{\theta_{o}}$- probability one.
\end{proposition}
\begin{proof}  For  economy of  space let us set $ C\left(l, n_{o}\right) :={\rm dim} (\Theta^{(l)}) \ln \frac{n_{o}}{8\pi}$. There  are two cases with distinct arguments.  
\begin{description}
\item[$l  > l_{o}$:] When  $ P_{\theta_{o}} \in \mathbb{M}^{\left(l_{o} \right)}$,  then $ P_{\theta_{o}} \in \mathbb{M}^{\left(l \right)}$ by the  nesting property.   It holds thus  that 
 $\theta_{o} \in \Theta^{\left(l_{o}\right)}$  implies  $\theta_{o} \in \Theta^{\left(l\right)}$ for every $ l$ such that  $l>l_{o}$, too. Let  now  $ \widehat{\theta}^{(l)}_{\rm JSD} = \widehat{\theta}^{(l)}_{\rm JSD} (\mathbf{D})$ be given in 
(\ref{minjsd}). 
As already stated  above,  \citet[Lemma 5.1, Proposition 5.2]{corremkos} shows that when   $P_{\theta_{o}} \in  \mathbb{M}^{\left(l\right)}$, then 
$$
 D_{\rm JS}\left( P_{\theta_{o}} ,P_{\widehat{\theta}^{(l)}_{\rm JSD}}\right) \rightarrow 0
$$
 $ P_{\theta_{o}}$  a.s., as $n_{o} \rightarrow +\infty$ for  all $l  > l_{o}$   and also for $l= l_{o}$. 
This means that there is some $n_{\epsilon}$ such that  for  $n_{o} >n_{\epsilon}$
\begin{equation}\label{sicjsdcomp}
{\rm SIC}_{\rm JSD} \left( \mathbb{M}_{p}^{(l_{o})}\right ) =2 n_{o} \epsilon  +  C\left(l_{o}, n_{o}\right) \leq  2 n_{o} \epsilon  +  C\left(l, n_{o}\right) .
\end{equation}
since   $ C\left(l_{o}, n_{o}\right) \leq   C\left(l, n_{o}\right) $   for every $l > l_{o}$, when   $n_{o} >8 \pi $.  
\item[$l  < l_{o}$:] Since   $ P_{\theta_{o}} \notin \mathbb{M}^{\left(l \right)}$ for $l  < l_{o}$, then  it holds 
for every $l  < l_{o}$ and  for a   $\delta>0$  defined  by the inefficiency  of  $  \mathbb{M}^{\left(l \right)}$  w.r.t.    $ P_{\theta_ {o}}$  that 
$$
\delta:= \min_{\theta \in \Theta^{(l)}}D^{1/2}_{\rm JS} \left( P_{\theta_ {o}},  P_{\theta}\right) \leq  
D^{1/2}_{\rm JS} \left( P_{\theta_ {o}},   P_{ \widehat{\theta}^{(l)}_{{\rm JSD}} (\mathbf{D}) }\right)
$$
with  $\widehat{\theta}^{(l)}_{\rm JSD} (\mathbf{D}) $ defined as in Equation~(\ref{minjsd}) for any $l  < l_{o}$.
In view  of  Proposition  \ref{sqrmetric} we can apply the triangle inequality in the right-hand side to the effect  that 
$$
D^{1/2}_{\rm JS} \left( P_{\theta_ {o}},   P_{ \widehat{\theta}^{(l)}_{{\rm JSD}} (\mathbf{D}) }\right) \leq  D^{1/2}_{\rm JS} \left( P_{\theta_ {o}},    \widehat{P}_{\mathbf{D}}\right) + 
D^{1/2}_{\rm JS} \left( \widehat{P}_{\mathbf{D}},   P_{ \widehat{\theta}^{(l)}_{{\rm JSD}} (\mathbf{D}) }\right). 
$$
Hence 
$$
D_{\rm JS} \left( P_{\theta_ {o}},   P_{ \widehat{\theta}^{(l)}_{{\rm JSD}} (\mathbf{D}) }\right) \leq 2  D_{\rm JS} \left( P_{\theta_ {o}},    \widehat{P}_{\mathbf{D}}\right) + 2
D_{\rm JS} \left( \widehat{P}_{\mathbf{D}},   P_{ \widehat{\theta}^{(l)}_{{\rm JSD}} (\mathbf{D}) }\right). 
$$
As $\mathbf{D} \sim  P_{\theta_ {o}}$,   $ D^{1/2}_{\rm JS} \left( P_{\theta_ {o}},    \widehat{P}_{\mathbf{D}}\right) \rightarrow 0$, as $n_{o}$ increases to  $+\infty$, as shown in   \citet[Lemma 5.1, Proposition 5.2]{corremkos}.  It follows that   
$$
\delta^{2}  <  \liminf_{n_{o} \rightarrow +\infty}2D_{\rm JS}\left( \widehat{P}_{\mathbf{D}}\ ,P_{\widehat{\theta}^{(l)}_{\rm JSD}} \right).
$$
Hence  for    $l  < l_{o}$ 
\begin{equation}\label{unbded}
{\rm SIC}_{\rm JSD} \left( \mathbb{M}_{p}^{(l)}\right ) =2 n_{o} D_{\rm JS} \left( \widehat{P}_{ \mathbf{D}}, P_{\widehat{\theta}^{(l)}_{\rm JSD}}\right) + C\left(l, n_{o}\right)
\end{equation}
is a  function of  $n_{o}$ that will ultimately with   $P_{\theta_{o}}$ -probability one  exceed 
${\rm SIC}_{\rm JSD} \left( \mathbb{M}_{p}^{(l_{o})}\right )=2 n_{o} \epsilon  + C\left(l_{o}, n_{o}\right)$ established in the first case  of this proof.   In more detail,  suppose that $\epsilon  <  \delta^{2} $  and that  $n_{o}$   satisfies 
\begin{equation}\label{sicineq}
2 n_{o}\left[ D_{\rm JS} \left( \widehat{P}_{ \mathbf{D}}, P_{\widehat{\theta}^{(l)}_{\rm JSD}}\right)   -  \epsilon \right] > \ln \frac{n_{o}}{8\pi} \Delta, 
\end{equation}
where  $ \Delta := \left(  {\rm dim} (\Theta^{\left(l_{o} \right)})  -{\rm dim} (\Theta^{\left(l \right)}) \right) 
>0 $. 
As    $2n_{o}$ grows faster than  $ \ln \frac{n_{o}}{8\pi} $ when $n_{o}$ grows,  and  by Equation~(\ref{range})   the positive  factor multiplying  $2 n_{o}$ is bounded, there is an integer  $N$ such that  Equation~(\ref{sicineq}) holds for all   $n_{o} >N$ with $  P_{\theta_{o}} $ probability  one.   The inequality in Equation~(\ref{sicineq}) is  equivalent to 
$$
{\rm SIC}_{\rm JSD} \left( \mathbb{M}_{p}^{(l)}\right ) - {\rm SIC}_{\rm JSD} \left( \mathbb{M}_{p}^{(l_{o})}\right ) >0
$$
 \end{description}
Hence we have proved the  required consistency property. 
\end{proof}

\subsection{JSD-Razor Rule, when the True Distribution is not Covered   by the Models  }\label{sciconsisten2}
We shall apply  the following result, which is valid for  any $P_{o}$, inside or outside the models.  
\begin{proposition}\label{misspec11}
 Assume that Equation~(\ref{posass}) holds for
$P_{o} \in  \mathbb{P}$ and for  any $P_{\theta} \in \mathbb{M}_{p}$. Let  $ {\mathbf D}= (D_{1}, \ldots, D_{n_{o}}) $ be    an i.i.d. $n_{o}$-sample  $\sim$ $P_{o}$. Then  it holds that 
\begin{equation}\label{misspec1}
\lim_{n_{o} \rightarrow +\infty} D_{\rm JS} \left( \widehat{P}_{\mathbf D}, P_{\theta}\right) =  D_{\rm JS} \left( P_{o}, P_{\theta}\right). 
\end{equation}
$P_{o}$-a.s.. 
\end{proposition}
The proof is found  in  \citet{corremkos2}.  Proposition \ref{misspec11}  is next   applied to study  of model choice by minimization of $ {\rm SIC}_{\rm JSD} \left( \mathbb{M}_{p}^{(l)}\right )$ when  there are 
$L$ nested models  $\mathbb{M}_{p}^{(l)}$ such that   $P_{o}$ is outside of  $\mathbb{M}_{p}^{(l)}$  for every $l$.  Of course, Lemma \ref{hessenapprsats} used in derivation of  ${\rm SIC}_{\rm JSD}$   requires  that   $P_{o}  \in \mathbb{M}_{p}$.
The result below is perhaps  in spite of this a natural extension.  Here $ D_{\rm JS} \left( P_{o}, P_{\theta^{(l)}}\right) >0$ for every  $ P_{\theta^{(l)}} \in \mathbb{M}_{p}^{(l)}$  for every $l$.
\begin{lemma}\label{deltaast}
The  models $\mathbb{M}_{p}^{(l)}$, $l=1,  \ldots, L$ are  nested    
$$
\Theta^{(1)} \subseteq \Theta^{(2)} \subseteq  \ldots \subseteq  \Theta^{(L)}
$$
and $ {\rm dim}\left (\Theta^{(1)}  \right))\leq  {\rm dim} \left(\Theta^{(2)}\right) \leq \ldots  \leq {\rm dim} \left(\Theta^{(L)} \right)( <k)$.
Suppose $P_{o} \notin\mathbb{M}_{p}^{(l)}$ for every $l$.
Let  $ {\mathbf D}= (D_{1}, \ldots, D_{n_{o}}) $ be    an i.i.d. $n_{o}$-sample  $\sim$ $P_{o}$. The assumption in Equation~(\ref{posass}) holds for
$P_{o} \in  \mathbb{P}$. 
  Let  us assume that there is  $ P_{\theta^{(l^{\ast})}} \in  \mathbb{M}_{p}^{(l^{\ast})}$   such that 
\begin{equation}\label{mindist3}
0 <  \delta^{\ast}:=  D^{1/2}_{\rm JS} \left( P_{o},  P_{\theta^{(l^{\ast})}}\right) =\min_{1 \leq l \leq L}  \min_{\theta^{(l)} \in \Theta^{(l)}} D^{1/2}_{\rm JS} \left( P_{o},  P_{\theta^{(l)}}\right) 
\end{equation}
Then, for every $l \geq  l^{\ast}$, as  $n_{o} \rightarrow +\infty$, 
\begin{equation}\label{misspclimit}
D_{\rm JS} \left( \widehat{P}_{\mathbf{D}},   P_{ \widehat{\theta}^{(l)}_{{\rm JSD}} (\mathbf{D}) }\right) \rightarrow  ( \delta^{\ast})^{2}
\end{equation}
with   $ P_{o}$- probability one,   where $\widehat{\theta}^{(l)}_{\rm JSD}$  is computed  in Equation~(\ref{minjsd}).
 \end{lemma}
\begin{proof}  It holds by Equation~(\ref{mindist3}) for every $l \geq l^{\ast}$ that 
\begin{equation}\label{mindist4}
0 <  \delta^{\ast} \leq D^{1/2}_{\rm JS} \left( P_{o},   P_{ \widehat{\theta}^{(l)}_{{\rm JSD}} (\mathbf{D}) }\right)
\end{equation}
By   the triangle inequality  justified  by Proposition \ref{sqrmetric}  we have 
\begin{equation}\label{mindist45}
D^{1/2}_{\rm JS} \left( P_{o}, P_{\widehat{\theta}^{(l)}_{\rm JSD}}\right)   \leq    D^{1/2}_{\rm JS} \left(P_{o},  \widehat{P}_{ \mathbf{D}}\right) +  D^{1/2}_{\rm JS} \left( \widehat{P}_{ \mathbf{D}}, P_{\widehat{\theta}^{(l)}_{\rm JSD}}\right).
\end{equation}
By  definition of    $\widehat{\theta}^{(l)}_{\rm JSD}$ and since the  models are nested  and  $l \geq  l^{\ast}$,  we bound the right hand side upwards by  
\begin{equation}\label{hakbd}
\leq   D^{1/2}_{\rm JS} \left(P_{o},  \widehat{P}_{ \mathbf{D}}\right) +  D^{1/2}_{\rm JS} \left( \widehat{P}_{ \mathbf{D}},   P_{\theta^{(l^{\ast})}}\right).
 \end{equation}
 The proof  of Theorem 17 in \citet{corremkos2}   can be used ad verbatim to show that
\begin{equation}\label{totvarber110} 
 D^{1/2}_{\rm JS} \left(P_{o},  \widehat{P}_{ \mathbf{D}}\right)  \rightarrow 0, \text{as $n_{o} \rightarrow +\infty$}   
 \end{equation}
 $P_{o}$-a.s..  In addition,  Proposition \ref{misspec11} entails that 
\begin{equation}\label{totvarber1100}
D^{1/2}_{\rm JS} \left( \widehat{P}_{ \mathbf{D}},   P_{\theta^{(l^{\ast})}}\right)  \rightarrow  D^{1/2}_{\rm JS} \left( P_{o},   P_{\theta^{(l^{\ast})}}\right), \text{as $n_{o} \rightarrow +\infty$.}
 \end{equation}
$P_{o}$-a.s..  In view of  
Equations~(\ref{mindist4})--(\ref{totvarber1100}), we have  
$$
0 <  \delta^{\ast} \leq   D^{1/2}_{\rm JS} \left(P_{o},  \widehat{P}_{ \mathbf{D}}\right) +  D^{1/2}_{\rm JS} \left( \widehat{P}_{ \mathbf{D}}, P_{\widehat{\theta}^{(l)}_{\rm JSD}}\right)
$$
$$
\leq  \lim_{n_{o} \rightarrow +\infty}   D^{1/2}_{\rm JS} \left(P_{o},  \widehat{P}_{ \mathbf{D}}\right) +  \lim_{n_{o} \rightarrow +\infty}   
D^{1/2}_{\rm JS} \left( \widehat{P}_{ \mathbf{D}},   P_{\theta^{(l^{\ast})}}\right) 
$$
$$
=0 +  D^{1/2}_{\rm JS} \left( P_{o},   P_{\theta^{(l^{\ast})}}\right)= \delta^{\ast}.
$$
Hence the claim  in Equation~(\ref{misspclimit}) is established.  \end{proof}

\begin{proposition}
Under the assumptions  of Lemma  \ref{deltaast}, it holds for  large $n_{o}$  and  every $l$ 
\begin{equation}\label{jsdkonsist2}
{\rm SIC}_{\rm JSD} \left( \mathbb{M}_{p}^{\left(l^{\ast}\right)}\right ) \leq {\rm SIC}_{\rm JSD} \left( \mathbb{M}_{p}^{\left(l\right)}\right ) \quad \text{for every  $l \neq  l_{o}$ }
\end{equation}
with   $ P_{o}$- probability one.
\end{proposition}
\begin{proof} 
\begin{description}
\item[$l  > l^{\ast}$:]  By Lemma \ref{deltaast} and replacement of $P_{\theta_{o}}$ with  $P_{o}$,  the proof of    Proposition  \ref{cons1} can be modified  to entail   the  statement  in Equation~(\ref{sicjsdcomp})  in the form that there is 
 some $n_{\epsilon}$ such that  for  $n_{o} >n_{\epsilon}$
\begin{equation}\label{sicjsdcomp22}
{\rm SIC}_{\rm JSD} \left( \mathbb{M}_{p}^{(l_{o})}\right ) =2 n_{o}  ( \delta^{\ast})^{2}  +  C\left(l^{\ast}, n_{o}\right) \leq  2 n_{o}  
(\delta^{\ast})^{2}  +  C\left(l, n_{o}\right) .
\end{equation}
since   $ C\left(l^{\ast}, n_{o}\right) \leq   C\left(l, n_{o}\right) $   for every $l >   l^{\ast}$, when   $n_{o} >8 \pi $.  
\item[$ l< l^{\ast}$:] Since   $ P_{o} \notin \mathbb{M}^{\left(l \right)}$ for $l  <   l^{\ast}$, then  it holds 
for every $l  < l_{o}$ and  for     $\delta^{\ast}>0$ in Equation~(\ref{mindist3})  that 
$$
0 <  \delta^{\ast}\leq   \min_{\theta \in \Theta^{(l)}}D^{1/2}_{\rm JS} \left( P_{o},  P_{\theta}\right) \leq  
D^{1/2}_{\rm JS} \left( P_{o},   P_{ \widehat{\theta}^{(l)}_{{\rm JSD}} (\mathbf{D}) }\right).
$$
It follows modifying  the second case proof  of  Proposition  \ref{cons1} that   
$$
( \delta^{\ast})^{2}  <  \liminf_{n_{o} \rightarrow +\infty}2D_{\rm JS}\left( \widehat{P}_{\mathbf{D}}\ ,P_{\widehat{\theta}^{(l)}_{\rm JSD}} \right).
$$
Hence  for    $ l< l^{\ast}$ 
\begin{equation}\label{unbded22}
{\rm SIC}_{\rm JSD} \left( \mathbb{M}_{p}^{(l)}\right ) =2 n_{o} D_{\rm JS} \left( \widehat{P}_{ \mathbf{D}}, P_{\widehat{\theta}^{(l)}_{\rm JSD}}\right) + C\left(l, n_{o}\right).
\end{equation}
is a  function of  $n_{o}$ that will ultimately with   $P_{o}$ -probability one  exceed 
${\rm SIC}_{\rm JSD} \left( \mathbb{M}_{p}^{(l^{\ast})}\right )=2 n_{o} ( \delta^{\ast})^{2}  + C\left(l^{\ast}, n_{o}\right)$ established in the first case  of this proof.   The rest of the proof is as in   Proposition  \ref{cons1}.
 \end{description}

\end{proof}

\section{Simulation Experiments}\label{sec:experiments}

\renewcommand{\arraystretch}{1.2}

\newcommand{\lx}{\lambda^X} 
\newcommand{\ly}{\lambda^Y} 
\newcommand{\lxy}{\lambda^{XY}} 

\newcommand{\minjsd}{D_{\rm JS} \left( \widehat{P}_{ \mathbf{D}}, P_{\widehat{\theta}^{(l)}_{\rm JSD}}\right)}

The experiments carried out in this work are aimed to evaluate the proposed SIC-JSD rule and its simulator-based approximation SIC-BOLFI.
We evaluate the model selection rules using simulation experiments where the true model is known, and test simulator-based model selection in a real task studied in previous work \citep{corander2017frequency}.
Section~\ref{sec:setup} reviews the model selection rules used in the experiments while Sections \ref{sec:experiment1}--\ref{sec:experiment3} present the experiments and results in more detail.

\subsection{Methods}\label{sec:setup}

The model selection experiments  considered in this work compare candidate models $\mathbb M_p^{(l)}$ fitted to observed data.
We run experiments with simulated observations to evaluate how the proposed model selection rule and its simulator-based approximation behave in different test conditions and when we increase the observed data set size $\nobs$.
While the present work focuses on parameter estimation and model selection in simulator-based models with intractable likelihoods, the simulation experiments in Section \ref{sec:experiment1}--\ref{sec:experiment2} were carried out with models where the mapping between model parameters $\param$ and observation probabilities $P_\param$ is available.
These experiments allowed us to evaluate model selection based on SIC (Equation \ref{scicrit}) and SIC-JSD (Equation \ref{sicjsdo}).


When the mapping between model parameters and observation probabilities is complicated or unknown, parameter estimation and model selection are carried out based on comparison between observed and simulated data.
The present experiments use BOLFI \citep{gutmann2016bayesian} to find parameter values that minimize the expected JSD between observed and simulated data. 
The experiments were carried out with the implementation available in ELFI \citep{lintusaari2018elfi}.
Gaussian process regression with a normal likelihood and squared exponential kernel was used to model the dependencies between simulator parameters and expected JSD between observed and simulated data in BOLFI, and the parameter values used in the simulations were selected based on the lower confidence bound acquisition rule (Section \ref{sec:experiment1}--\ref{sec:experiment2}) or the maximum variance acquisition rule (Section \ref{sec:experiment3}).
The candidate models studied in Section \ref{sec:experiment1} were fitted using 1000 iterations and the candidate models Section \ref{sec:experiment2} and \ref{sec:experiment3} using 2000 iterations.

\subsection{Experiment 1}\label{sec:experiment1}


This experiment studies model selection using the nested model introduced in Example \ref{multicateg2}.
We simulate 100 observation sets with $\nobs=100$ and $\nobs=1000$ samples using model $\mathbb M_p^{(0)}$ which corresponds to $\theta=(0,0)$, model $\mathbb M_p^{(1)}$ with $\theta=(0.2, 0)$ and $\theta=(0.7,0)$, and model $\mathbb M_p^{(2)}$ with $\theta=(0.2,0.2)$, $\theta=(0.7,0.2)$, $\theta=(0.2,0.7)$, and $\theta=(0.7,0.7)$.
The corresponding category probabilities are visualized in Figure~\ref{fig:event_probabilities_example_22}.
\begin{figure}
    \centering
    \includegraphics[width=0.9\textwidth]{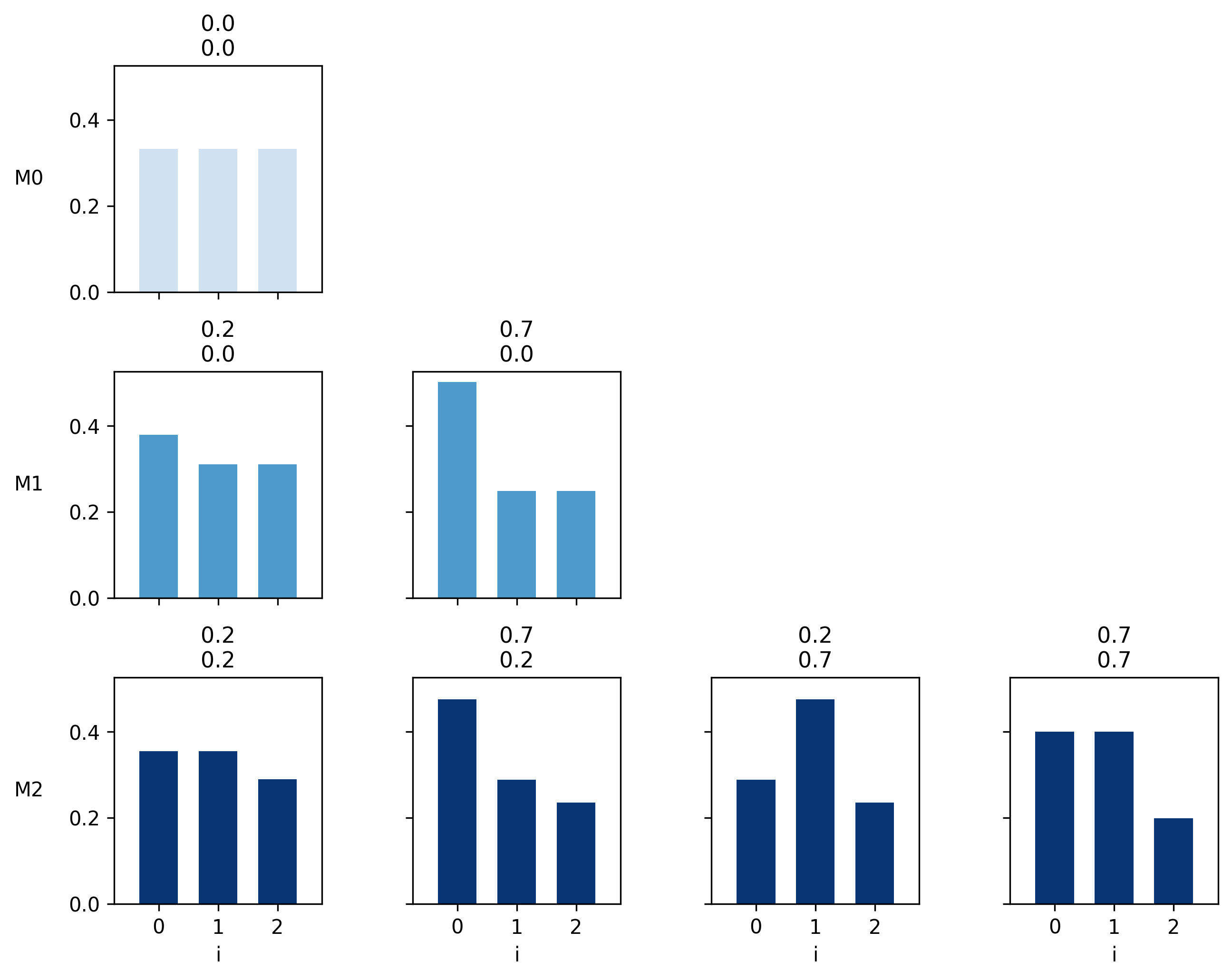}
    \caption{Models and category probabilities used to simulate the observation sets used in the model selection experiment carried out with the nested model in Example \ref{multicateg2}.}
    \label{fig:event_probabilities_example_22}
\end{figure}


The candidate models $\mathbb M_p^{(0)}$, $\mathbb M_p^{(1)}$, and $\mathbb M_p^{(2)}$ are each fitted to the simulated observation sets and model selection is carried out based on SIC, SIC-JSD, and SIC-BOLFI.
The model selection results calculated based on observation sets with $\nobs=100$ observations are presented in Figure~\ref{fig:modelval_example_22_100} and the results calculated based on observation sets with $\nobs=1000$ observations in Figure~\ref{fig:modelval_example_22_1000}.
We observe that when $\nobs=100$, all model selection rules favor $\mathbb M_p^{(0)}$ in test conditions where the true model parameters have low values, and $\mathbb M_p^{(1)}$ in test conditions where the second parameter alone has a low value.
However the model selection rules do not agree in all experiments.
SIC chooses $\mathbb M_p^{(0)}$ or $\mathbb M_p^{(1)}$ over the true model in more experiments than SIC-JSD or SIC-BOLFI, while SIC-JSD and SIC-BOLFI choose $\mathbb M_p^{(2)}$ over the true model in more experiments than SIC. 
%
%
This difference disappears when $\nobs=1000$, and all model selection rules also choose the true model in more experiments.

\begin{figure}
    \centering
    \includegraphics[width=\textwidth]{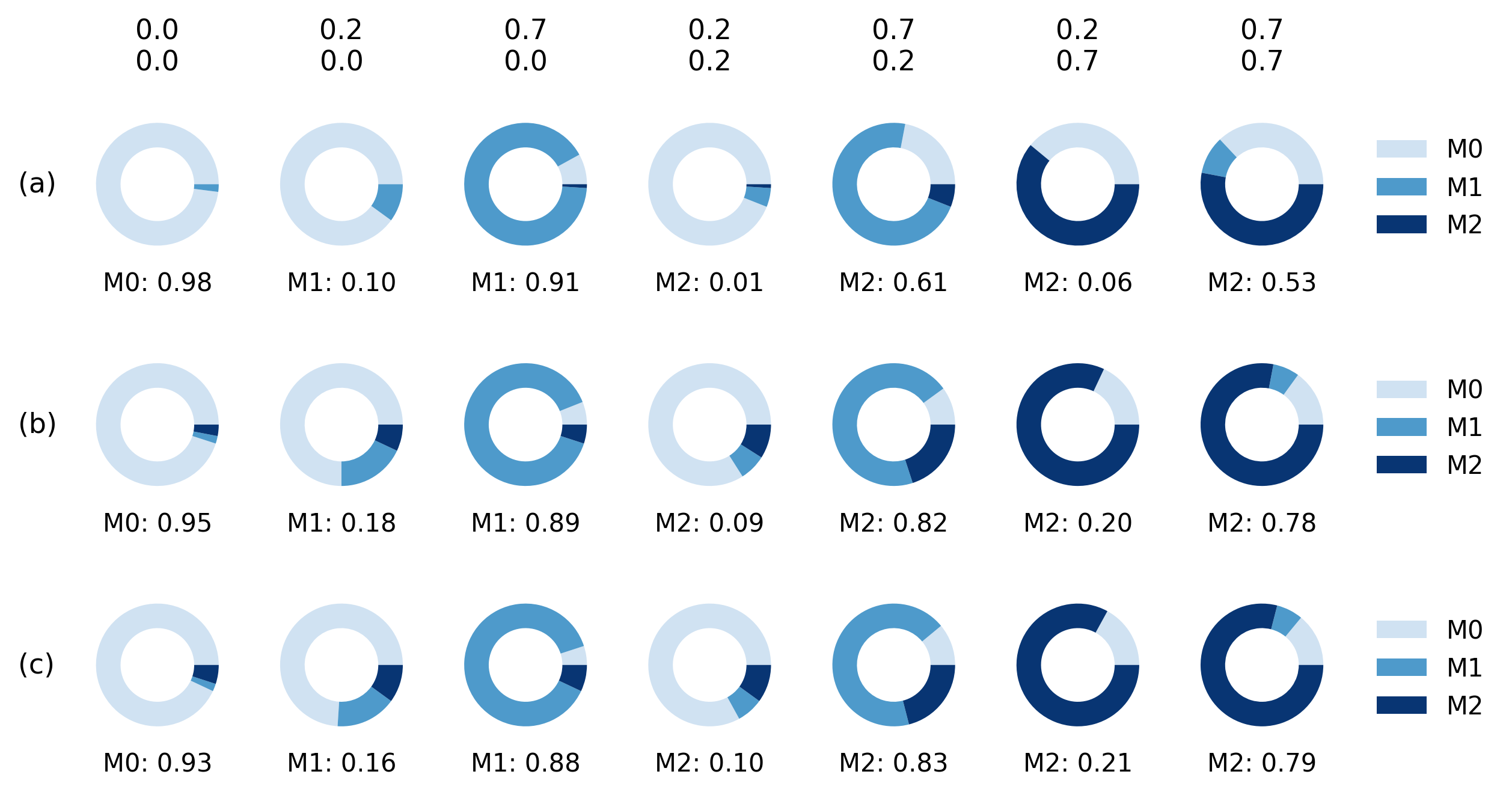}
    \caption{Model selection results calculated based on 100 observation sets with $\nobs=100$ observations simulated with the models in Figure \ref{fig:event_probabilities_example_22}. Model selection rules used in the experiment are (a) SIC, (b) SIC-JSD, and (c) SIC-BOLFI.}
    \label{fig:modelval_example_22_100}
\end{figure}

\begin{figure}
    \centering
    \includegraphics[width=\textwidth]{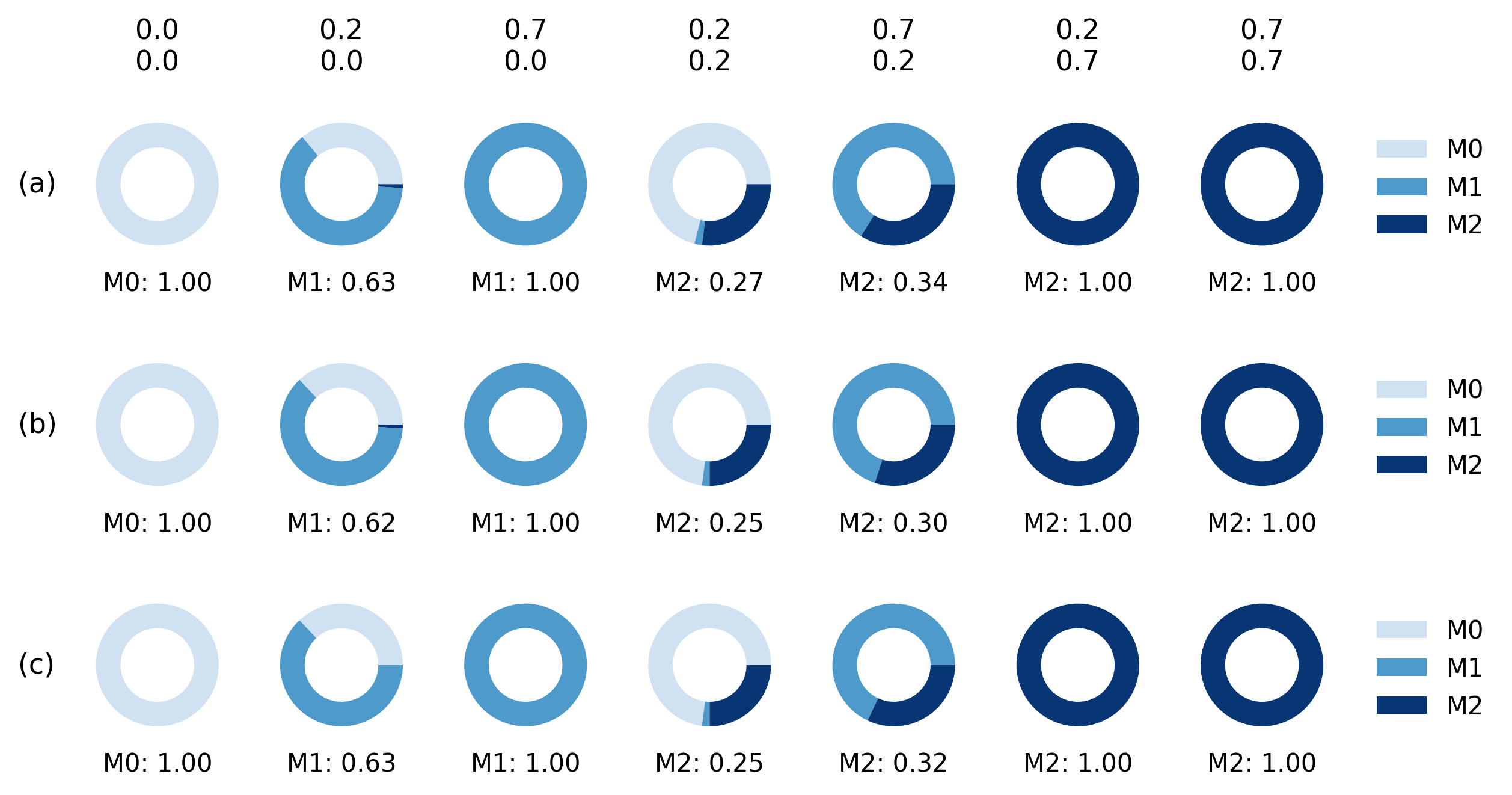}
    \caption{Model selection results calculated based on 100 observation sets with $\nobs=1000$ observations simulated with the models in Figure \ref{fig:event_probabilities_example_22}. Model selection rules used in the experiment are (a) SIC, (b) SIC-JSD, and (c) SIC-BOLFI.}
    \label{fig:modelval_example_22_1000}
\end{figure}

\subsection{Experiment 2}\label{sec:experiment2}


The second experiment studies model selection between log-linear models that describe the association and interaction patterns between two categorical random variables.
The counts in a two-way table are modeled as a sample from a multinomial distribution with $k=4$ categories and expected observation counts $\mu_\icls$ modeled as
\begin{equation}
\log(\mu_\icls) = \lambda+ X_\icls\lx+ Y_\icls\ly+X_\icls Y_\icls\lxy, \quad \icls = 1,2,3,4,
\end{equation}
where $\icls$ indexes the categories and $X_i$ and $Y_i$ are effect-coded variables that take values 1 or -1 as indicated in Table~\ref{tab:effectcode}.
\begin{table}
    \centering
    \begin{tabular}{ccccc}
    $\icls$ & 1 & 2 & 3 & 4  \\\hline
    $X_i$ & 1 & 1 & -1 & -1 \\ 
    $Y_i$ & 1 & -1 & 1 & -1 \\
    \end{tabular}
    \caption{\label{tab:effectcode}Effect coding in the log-linear example.}
\end{table}
The model parameters $\lx$ and $\ly$ then encode expected difference in the proportion between 1 and -1 values in variables $X$ and $Y$, and the parameter $\lxy$ encodes possible association between the two variable values.
Finally the constant $\lambda$ is calculated based on the other parameter values and sample size $\nsms$ so that the sum over expected counts equals $\nsms$.


We run experiments with a two-parameter model $\mathbb M_p^{(2)}$ where $\lxy=0$ and the model parameters $\param=(\lx,\ly)$ and a saturated or three-parameter model $\mathbb M_p^{(3)}$ where the model parameters $\param=(\lx,\ly,\lxy)$.
The observation sets used in the experiments are simulated with parameter values selected as follows.
The model parameters $\lx$ and $\ly$ are first associated with 100 values selected at random within $[-1, 1]\times[-1, 1]$ while the interaction parameter is associated with 11 values selected at interval $0.1$ between $-0.5$ and $0.5$.
We then combine each $\lxy$ value with the $\lx$ and $\ly$ values to create $11\times 100$ parameter combinations and use each combination to simulate observation sets with $\nobs=100$ and $\nobs=1000$ samples.


Model selection results between $\mathbb M_p^{(2)}$ and $\mathbb M_p^{(3)}$ are presented in Table \ref{tab:modelval_loglin}.
%
%
We observe that SIC-JSD and SIC-BOLFI choose $\mathbb M_p^{(3)}$ over $\mathbb M_p^{(2)}$ in more experiments than SIC when $\nobs=100$, but the differences disappear when $\nobs=1000$.
Comparison between the model selection rates in experiments with $\nobs=100$ and $\nobs=1000$ also indicates that all model selection rules choose the true model in more experiments when $\nobs=1000$.

\begin{table}
    \centering
    \begin{tabular}{ll}
    (a) &
    \begin{tabular}{lccccccccccc}
    $\lxy$ & -0.5 & -0.4 & -0.3 & -0.2 & -0.1 & 0.0 & 0.1 & 0.2 & 0.3 & 0.4 & 0.5\\\hline
    SIC & 0.94 & 0.75 & 0.53 & 0.33 & 0.11 & 0.01 & 0.0 8& 0.42 & 0.59 & 0.77 & 0.93\\
    SIC-JSD & 0.99 & 0.87 & 0.67 & 0.44 & 0.20 & 0.10 & 0.17 & 0.55 & 0.76 & 0.92 & 0.96\\
    SIC-BOLFI & 0.99 & 0.89 & 0.68 & 0.46 & 0.24 & 0.16 & 0.22 & 0.58 & 0.77 & 0.94 & 0.96\\
    \end{tabular}\\
    & \\
    (b) &
    \begin{tabular}{lccccccccccc}
    $\lxy$ & -0.5 & -0.4 & -0.3 & -0.2 & -0.1 & 0.0 & 0.1 & 0.2 & 0.3 & 0.4 & 0.5\\\hline
    SIC & 1.00 & 1.00 & 1.00 & 0.96 & 0.41 & 0.00 & 0.43 & 0.95 & 1.00 & 1.00 & 1.00\\
    SIC-JSD & 1.00 & 1.00 & 1.00 & 0.95 & 0.39 & 0.00 & 0.41 & 0.95 & 1.00 & 1.00 & 1.00\\
    SIC-BOLFI & 1.00 & 1.00 & 1.00 & 0.95 & 0.39 & 0.00 & 0.39 & 0.95 & 1.00 & 1.00 & 1.00\\
    \end{tabular}
    \end{tabular}
    \caption{$\mathbb M_p^{(3)}$ selection rate calculated based on 100 observation sets with (a) $\nobs=100$ or (b) $\nobs=1000$ observations simulated with the log-linear model.}
    \label{tab:modelval_loglin}
\end{table}

\subsection{Experiment 3}\label{sec:experiment3}


The last experiment is carried out with the simulator models and data used in previous work \citep{corander2017frequency}.
\citet{corander2017frequency} studied negative frequency-dependent selection (NFDS) in genotype frequencies in post-vaccine pneumococcal populations.
We replicate a comparison between three simulators that model the evolution in genotype frequencies as a discrete-time process where the population at time $t$ is sampled with replacement from population at time $t-1$ using reproduction probabilities calculated based on simulator parameters $\param$.
The candidate models are nested and include the neutral model, the homogeneous-rate multilocus NFDS model, and the heterogeneous-rate multilocus NFDS model.


The models considered in this experiment calculate the reproduction probabilities based on 2--5 parameters as follows.
The neutral model $\mathbb M_p^{(2)}$ takes into account the migration rate $m$ and the vaccine selection strength $v$.
These capture the negative pressure due to migration into population and the negative selection pressure on vaccine-type isolates.
The other model variants extend the neutral model to take into account how isolates with rare variations in their accessory genome could experience positive selection pressure under NFDS.
The difference between the two models is how variations that occur in different loci contribute in the selection pressure.
The homogeneous-rate multilocus NFDS model $\mathbb M_p^{(3)}$ associates all variations with the maximal selection pressure $\sigma_f$ while the heterogeneous-rate multilocus NFDS model $\mathbb M_p^{(5)}$ associates variations in some loci with a weaker selection pressure $\sigma_w$.
The proportion of loci under stronger NFDS is captured with the model parameter $p_f$.


The candidate models are fitted to data that was collected to follow how vaccination affected the pneumococcal population in Massachusetts \citep{mass}.
The data set used in this experiment includes a pre-vaccination ($t=0$) sample with 133 isolates and two post-vaccination samples with 203 isolates collected at $t=36$ and 280 isolates collected at $t=72$.
The isolates have been divided into 41 sequence clusters and typed as vaccine or non-vaccine type as discussed in previous work \citep{corander2017frequency}.
The data set is visualized in Figure~\ref{fig:mass_data}.
\begin{figure}
    \centering
    \includegraphics[width=0.9\textwidth]{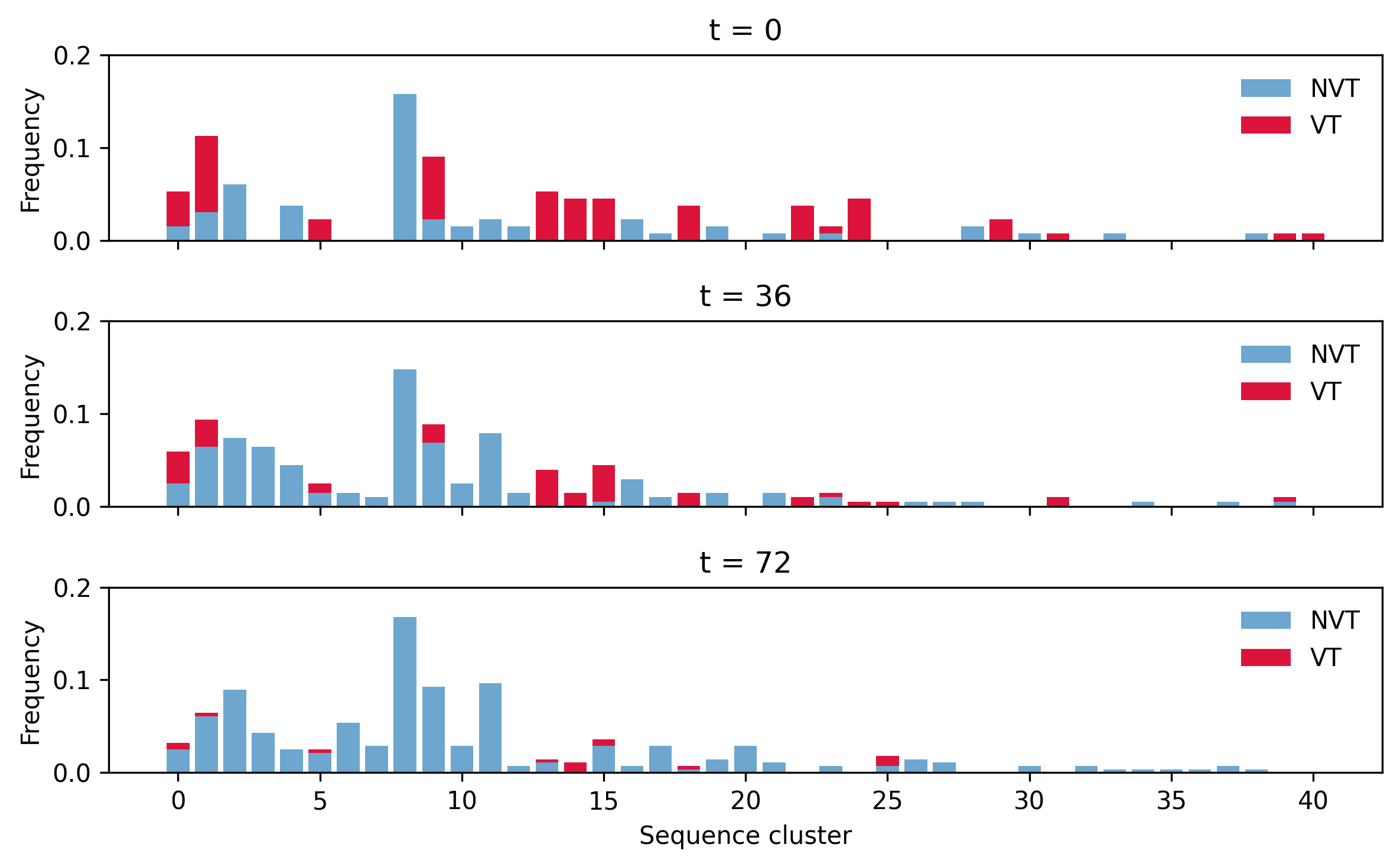}
    \caption{Sequence cluster frequencies in the observed population data. The observations are divided into non-vaccine type (NVT) and vaccine type (VT).}
    \label{fig:mass_data}
\end{figure}
The negative selection pressure on vaccine-type isolates is observed as a reduction in the vaccine-type isolates over time.
In addition we observe a change in the relative frequencies between sequence clusters represented in the population that is non-vaccine tye.
This could indicate that the vaccine resulted in a positive selection pressure on isolates in certain sequence clusters.


Parameter estimation is carried out as proposed in the previous work.
We sample the pre-vaccine ($t=0$) data to initialize the simulated population in each candidate model and use BOLFI to find the parameter estimates that minimize the expected JSD between observed and simulated data sets at $t=36$ and $t=72$.
Hence the sample size used in parameter estimation and model selection $\nobs=483$.
Parameter estimation is carried out over $\ln(m)\in [-7, -1.6]$, $\ln(v) \in [-7, -0.7]$, $\ln(\sigma_f)\in [-7, -1.6]$, $\ln \sigma_w\in [-7, -1.9]$, and $p_f\in [0, 1]$ with the constraint $\sigma_f>\sigma_w$, and the model selection criterion is evaluated based on the average over the expected JSD at $t=36$ and $t=72$.
The parameter estimation and model selection are replicated 100 times to capture possible random variation between parameter estimates and model fit evaluated with BOLFI.


The model selection results reported in Table \ref{tab:modelval_mass} indicate that SIC-BOLFI chooses the heterogeneous multilocus NFDS model $\mathbb M_p^{(5)}$ over the neutral model $\mathbb M_p^{(2)}$ or the homogeneous multilocus NFDS model $\mathbb M_p^{(3)}$.
\begin{table}
    \centering
    \begin{tabular}{cccccccc}
        Model & Selection rate & JSD & $\hat m$ & $\hat v$ & $\hat\sigma_f$ & $\hat\sigma_w$ & $\hat p_f$ \\\hline
        $\mathbb M_p^{(2)}$ & 0.00 & 0.23 & 0.007 & 0.037\\
        $\mathbb M_p^{(3)}$ & 0.00 & 0.20 & 0.006 & 0.073 & 0.008\\
        $\mathbb M_p^{(5)}$ & 1.00 & 0.14 & 0.005 & 0.088 & 0.114 & 0.002 & 0.372 \\
    \end{tabular}
    \caption{Model selection rate and median parameter estimation results calculated based on 100 BOLFI experiments with the pneumococcal population models and data.
    }
    \label{tab:modelval_mass}
\end{table}
%
The outcome seems reasonable since $\mathbb M_p^{(5)}$ introduces a notable improvement in the model fit.
However we also check that the model selection rule is not biased towards $\mathbb M_p^{(5)}$ by running an additional simulation experiment.
We simulate 100 observation sets with 250 isolates sampled at $t=36$ and $t=72$ ($\nobs=500$) using $\mathbb M_p^{(2)}$ with $m=0.007$ and $v=0.050$ and $\mathbb M_p^{(3)}$ with $m=0.007$, $v=0.050$, and $\sigma_f=0.007$.
The model selection results reported in Table \ref{tab:modelval_simu} show that the true model is selected in most experiments.
%

\begin{table}
    \centering
    \begin{tabular}{ll}
    (a) &
    \begin{tabular}{cccccccc}
        Model & Selection rate & JSD \\\hline 
        $\mathbb M_p^{(2)}$ & 0.84 & 0.11 \\
        $\mathbb M_p^{(3)}$ & 0.10 & 0.11 \\
        $\mathbb M_p^{(5)}$ & 0.06 & 0.11 \\
    \end{tabular} \\
    & \\
    (b) & 
    \begin{tabular}{cccccccc}
        Model & Selection rate & JSD \\\hline 
        $\mathbb M_p^{(2)}$ & 0.00 & 0.14 \\
        $\mathbb M_p^{(3)}$ & 0.99 & 0.12 \\
        $\mathbb M_p^{(5)}$ & 0.01 & 0.12 \\
    \end{tabular}
    \end{tabular}
    \caption{Model selection rate and median minimum expected JSD calculated based on 100 observation sets simulated with the pneumococcal population models (a) $\mathbb M_p^{(2)}$ and (b) $\mathbb M_p^{(3)}$.
    }
    \label{tab:modelval_simu}
\end{table}

\section{Discussion and Conclusions}

Model choice as a statistical problem has a rich history in computer science, inspired in particular by information theory, and in statistics, where the major innovations have been founded on Bayesian thinking, which provides an intrinsic solution to the need to penalize more complex models by the prior distribution of the model parameters. However, literature on model choice for the likelihood-free inference setting is scarce, which is understandable since the vast majority of existing model scoring criteria use the likelihood in one way or another to measure the fidelity of a model as a data representation. 

The approach introduced here (JSD-Razor) was inspired by the Occam$^{,}$s  Razor as developed  by Balasubramanian and co-workers   in  \citet{balasubramanian2005mdl}, \citet{balasubramanian1996geometric} and  \citet{myung2000counting}, as well as by our previous work on the asymptotics of likelihood-free parameter inference under JSD  \citep{corremkos}. To the best of our knowledge, this is the first information-theoretic model scoring criterion introduced for simulator-based likelihood-free modeling setting. We anticipate that there are multiple opportunities for future developments in this area that can broaden the applicability of JSD-Razor to several other classes of models and spawn even more refined scoring criteria. For example, in model classes with continuous output, it would be possible to consider quantization to make JSD-Razor applicable, which raises several interesting questions related to the relative loss of information as a function quantization scheme and the sample size.  

\section*{Acknowledgments} 

The authors wish to acknowledge CSC $-$ IT Center for Science, Finland, for computational resources.
J.C. and U.R. are supported by  ERC grant 742158   and T.K.  is supported by  FCAI (=Finnish Center for Artificial Intelligence).

\newpage 
\appendix 
\numberwithin{equation}{section}
\numberwithin{theorem}{section}
\numberwithin{example}{section}

\section{Expressions and   Bounds for JSD }\label{app_a} 
Let us next define the   metric $ || .||_{2}$ on $\mathbb{P}  \times \mathbb{P}$ by  means of the norm on 
$\mathbf{R}^{k}$ as
\begin{equation}\label{birchlemma0}
 || P-Q ||_{2}: =\sqrt{\sum_{x \in {\cal A}} \left( P(x)-Q(x) \right)^{2} }= || \triangle(P) -\triangle(Q )||_{2, \mathbf{R}^{k}}.
\end{equation}  
\begin{lemma}\label{birchbd}
For  any  $(P,Q) \in \mathbb{P}  \times \mathbb{P}$ 
\begin{equation}\label{jsdl2}
D_{\rm JS }^{1/2}(P,Q) \geq   \frac{\sqrt{2}}{4}  || P-Q ||_{2}.
\end{equation}
 \end{lemma}
\begin{proof} 
Since the square root is a concave function on $[0, +\infty)$, we have by Equation~(\ref{jssymm})  
\begin{equation}\label{sqrconc}
D_{\rm  JS}^{1/2}(P,Q) \geq    \frac{1}{2}  D_{\rm KL}^{1/2}( P, M) +  \frac{1}{2}  D_{\rm KL}^{1/2}( Q,  M).
 \end{equation} 
It is shown in \citet[Lemma 1 p.~819]{birch1964} that  for  any  $(P,Q) \in \mathbb{P}  \times \mathbb{P}$ 
\begin{equation}\label{birchlemma1}
D_{\rm KL}(P,Q) \geq \frac{1}{2} || P-Q ||_{2}^{2}.  
\end{equation} 
 Hence we bound  downwards in the right hand side of  Equation~(\ref{sqrconc}) by 
\begin{equation}\label{birchtklemma1}
\geq   \frac{1}{2 \sqrt{2}} \left(  || P-M ||_{2}  + || Q-M ||_{2} \right). 
\end{equation} 
Here  by Equation~(\ref{birchlemma0})
$$
 || P-M ||_{2} = \sqrt{\sum_{i=1}^{k} \left(p_{i} -   \frac{1}{2}(p_{i}+ q_{i}) \right)^{2}}=   \frac{1}{2} || P-Q ||_{2}, 
$$
and 
$$
 || Q-M ||_{2} = \    \frac{1}{2} || Q-P ||_{2}= \frac{1}{2} || P-Q ||_{2}.
$$
When we insert the last two equalities in Equation~(\ref{birchtklemma1}),  the  assertion in the lemma follows. 
\end{proof}
The inequality above gives   a minor observation  with 
$$
\int_{\Theta} e^{-2n_{o}  D_{\rm JS}\left(  \widehat{P}_{\mathbf{D}},  P_{\theta}\right)}p(\theta) d\theta \leq
 \int_{\Theta} e^{- \frac{n_{o}}{4} ||  \widehat{P}_{\mathbf{D}}- P_{\theta}||_{2}^{2}  }p(\theta) d\theta. 
$$
The right-hand integral contains formally  the non-normalized multivariate normal density  with    the  $d \times d $ unit matrix  multiplied by  $\frac{2}{n_{o}} $  as covariance matrix, thus a kind of  likelihood function for $\theta$.

\section{Differential Calculus for  the Fisher  Information Matrix of  $\mathbb{M}_{p}$ and the  Hessian  of  JSD}\label{birchmapmain} 

\subsection{Information Matrix}\label{fisherrefapp}

Next we check   Lemma \ref{infemma}. 
\begin{proof}
 By definition $I_{ij}(\theta) $ is  
\begin{equation}\label{ijfisher2}
I_{ij}(\theta)
= \sum_{x \in  {\cal A}} \left(\frac{\partial}{\partial \theta_{i}}  \ln P_{\theta}(x) \frac{\partial}{\partial \theta_{j}}  \ln P_{\theta}(x) \right) P_{\theta}(x)
\end{equation}  
By properties of  the Iverson bracket and Equation~(\ref{pariversion}) it follows readily that 
\begin{equation}\label{ijfisher3}
I_{ij}(\theta) = \sum_{s=1}^{k} \frac{\frac{\partial}{\partial \theta_{i}}p_{s}(\theta) \frac{\partial}{\partial \theta_{j}}p_{s}(\theta)}{p_{s}(\theta)}.
\end{equation} 
 We  have  by rules of matrix multiplication and  the Jacobian given in Equation~(\ref{afaktor})  that 
\begin{equation}\label{nablanot6}
A(\theta)^{T}A(\theta) = \sum_{s=1}^{k} \frac{p_{s}^{'}(\theta)^{T} p_{s}^{'}(\theta)}{p_{s}(\theta)},
\end{equation}
where $^{T}$ is the vector transpose  and each $p_{s}^{'}(\theta)^{T} p_{s}^{'}(\theta)$  is a $d \times d$-matrix. 
The  array at position $(i,j)$   in  this matrix is  by  Equation~(\ref{nablanot4})
 \begin{equation}\label{nablanot7}
 \left(p_{s}^{'}(\theta)^{T} p_{s}^{'}(\theta) \right)_{i,j}=  \frac{\partial}{\partial \theta_{i}}p_{s}(\theta) \frac{\partial}{\partial \theta_{j}}p_{s}(\theta).
\end{equation}
Hence the  array at position $(i,j)$   in  $A(\theta)^{T}A(\theta) $ is 
\begin{equation}\label{nablanot8}
 \left(A(\theta)^{T}A(\theta) \right)_{i,j}=  \sum_{s=1}^{k} \frac{ \frac{\partial}{\partial \theta_{i}}p_{s}(\theta) \frac{\partial}{\partial \theta_{j}}p_{s}(\theta) }{p_{s}(\theta)}.
\end{equation}
But a comparison  with  Equation~(\ref{ijfisher3}) and Equation~(\ref{nablanot2}) yields  the asserted formula. 
\end{proof}
\begin{example}\label{multicateg3}
Explicit expressions for the  Fisher  information are established for the nested families   in  Example \ref{multicateg2}.  
 
\begin{description}
\item[ (i) $\theta= \left(\theta_{1},  \theta_{2} \right)$]
The  $2 \times 2$ Fisher information matrix  $I(\theta)$ is computed  for Equation~(\ref{multicat12})  by the formula in Equation~(\ref{nablanot8}) and is found elementwise as 
\begin{eqnarray}\label{fishif3kat}
I_{1,1}(\theta)&=&  e^{ -3M_{2}(\theta)}  e^{ \theta_{1} } \left[  e^{ \theta_{1} }  +\left(1+ e^{ \theta_{2} }  \right )^{2}  + e^{ \theta_{1}+ \theta_{2} }    \right],
 \\
I_{1,2}(\theta)&=& -e^{ -3M_{2}(\theta)} e^{ \theta_{1} + \theta_{2}} \left[ 1+ e^{ \theta_{1} }   + e^{ \theta_{2} }    \right]= I_{2,1}(\theta),
\nonumber \\
I_{2,2}(\theta)&= &e^{ -3M_{2}(\theta)} e^{ \theta_{2} } \left[  e^{ \theta_{2} }  +\left(1+ e^{ \theta_{1} }  \right )^{2}  + e^{ \theta_{1}+ \theta_{2} }    \right]. \nonumber 
\end{eqnarray}
The expression $ \sqrt{\det I(\theta)}$ is readily evaluated  and the integral  $ V\left( \Theta \right):=\int_{\Theta}  \sqrt{\det I(\theta)} d \theta$ in 
Equation~(\ref{notcoarse}) can  be computed, at least  numerically,  for  $\Theta$ such that   the integral exists.
\item[ (ii) $\theta= \left(\theta_{1}, 0\right)$] 
 The  scalar  Fisher information   $I(\theta)$ is computed  for Equation~(\ref{multicat13})  by the formula in Equation~(\ref{nablanot8}). 
\begin{equation}\label{fisinf2}
I(\theta) = 2 e^{ -3M_{1}(\theta)}  e^{ \theta_{1} } \left[ 2 + e^{ \theta_{1} }  \right].
\end{equation}
This agrees with $I_{1,1}(\theta)$ by setting   $\theta_{2}=0$   in  Equation~(\ref{fishif3kat}).  Since 
$M_{1}(\theta)= \ln \left( 2 +  e^{ \theta_{1}}      \right)$,  $I(\theta)$  in Equation~(\ref{fisinf2}) simplifies  to
$
I(\theta) =  2e^{ \theta_{1} }/\left(2 + e^{ \theta_{1}} \right)^{2}
$.
Thereby   the integral  in   Equation~(\ref{notcoarse})  becomes 
\begin{equation}\label{fisinfe2}
 V\left( \Theta \right)=\sqrt{2}\int_{\Theta} \frac{e^{ \theta_{1}/2 }}{ \left(2 + e^{ \theta_{1}} \right)}d\theta_{1}.
\end{equation}
When $\Theta=[a,b]$ we get by some changes of the variable of integration  that 
\begin{equation}\label{fisinfe3}
 V\left( [a,b] \right)=
2 \left[\arctan\left(\frac{e^{b/2 }}{\sqrt{2}}  \right)  - \arctan\left(\frac{e^{a/2 }}{\sqrt{2}}  \right) \right].
\end{equation}
Hence  $ 0 <  V\left( [a,b] \right) < 2 \pi$ for  $-\infty < a<b <+\infty$. 
  \item[ (iii) $\theta= \left(0, 0\right)$] The rule  in   Equation~(\ref{nablanot8})  is not  defined.    However,  from $\bf{(ii)}$,
$
 2e^{ 0 }/\left(2 + e^{ 0} \right)^{2}$ $  =2/9$.
\end{description}
\end{example}

\subsection{Hessian }\label{diffintvajdapist}
First  we prove Lemma \ref{birchloglike}.
\begin{proof}This is most  conveniently  done by differentiating   the generic expression 
\begin{equation}\label{jsdvectormap22}
\Phi\left(  {\bf p},  \theta   \right)= \sum_{i=1}^{k} p_{i}(\theta)  \phi \left( \frac{p_{i}}{p_{i}(\theta) }   \right), 
\end{equation} 
where $\phi$ is given in  Equation~(\ref{jsddiv}). 
For $j=1, \ldots, d$
\begin{equation}\label{partialq} 
\frac{\partial}{\partial \theta_{j}}\Phi\left(  {\bf p},  \theta \right)=  \sum_{i=1}^{k} \left[\phi \left( \frac{p_{i}}{p_{i}(\theta)}  \right)  - \frac{p_{i}}{{p_{i}}(\theta) }   \phi^{'}\left( \frac{p_{i}}{{p_{i}}(\theta) }   \right)\right] \frac{\partial}{\partial \theta_{j}}p_{i}(\theta).
\end{equation} 
 By Equation~(\ref{jsddivder4}) we obtain 
$$
\frac{\partial}{\partial \theta_{j}}\Phi\left(  {\bf p},  \theta \right)= -\frac{1}{2 } \sum_{i=1}^{k} \ln \left(\frac{\frac{1}{2 } \left( p_{i} +p_{i}(\theta)\right)}{ p_{i}(\theta)} \right)\frac{\partial}{\partial \theta_{j}}p_{i}(\theta),
$$
i.e. 
\begin{equation}\label{partialq25} 
\frac{\partial}{\partial \theta_{j}}\Phi\left(  {\bf p},  \theta \right)=\frac{1}{2 } \sum_{i=1}^{k} \ln \left(\frac{2 p_{i}(\theta)}{\left( p_{i} +p_{i}(\theta)\right)} \right)\frac{\partial}{\partial \theta_{j}}p_{i}(\theta).
\end{equation} 
By Equation~(\ref{jsddivder3}), the definition of the   the $k \times d$ Jacobian $J\left(\theta \right)$  in Equation~(\ref{nablanot3}) and rules of matrix multiplication the expression in  Equation~(\ref{gradth}) 
is established. \end{proof}


Next we prove Lemma \ref{jsdhessian}.
\begin{proof}  
 It is convenient to  compute by means of  Equation~(\ref{partialq}). 
For $l=1, \ldots, d$ we get by the chain  rule 
\begin{eqnarray}\label{partialq2}
\frac{\partial^{2}}{\partial \theta_{l}\partial \theta_{j}}\Phi\left(  {\bf p},  \theta \right) =  \frac{\partial}{\partial \theta_{l}}\sum_{i=1}^{k} \left[\phi \left( \frac{p_{i}}{p_{i}(\theta)}  \right)  - \frac{p_{i}}{{p_{i}}(\theta) }   \phi^{'}\left( \frac{p_{i}}{{p_{i}}(\theta) }   \right)\right] \frac{\partial}{\partial \theta_{j}}p_{i}(\theta) & &  \nonumber \\
& & \\
= \sum_{i=1}^{k} \left[\phi \left( \frac{p_{i}}{p_{i}(\theta)}  \right)  - \frac{p_{i}}{{p_{i}}(\theta) }   \phi^{'}\left( \frac{p_{i}}{{p_{i}}(\theta) }   \right)\right] \frac{\partial^{2}}{\partial \theta_{l} \partial \theta_{j}}p_{i}(\theta) & &   \nonumber \\
 + \sum_{i=1}^{k} \frac{\partial}{\partial \theta_{l}} \left[\phi \left( \frac{p_{i}}{p_{i}(\theta)}  \right)  - \frac{p_{i}}{{p_{i}}(\theta) }   \phi^{'}\left( \frac{p_{i}}{{p_{i}}(\theta) }   \right)\right] \frac{\partial}{\partial \theta_{j}}p_{i}(\theta). & &\nonumber  
\end{eqnarray}
As observed in the proof of  Lemma \ref{birchloglike}
\begin{equation}\label{obslema}
\sum_{i=1}^{k} \left[\phi \left( \frac{p_{i}}{p_{i}(\theta)}  \right)  - \frac{p_{i}}{{p_{i}}(\theta) }   \phi^{'}\left( \frac{p_{i}}{{p_{i}}(\theta) }   \right)\right] \frac{\partial^{2}}{\partial \theta_{l} \partial \theta_{j}}p_{i}(\theta) = 
\sum_{i=1}^{k}   \phi^{'}\left( \frac{p_{i}(\theta)}  {p_{i}}  \right) \frac{\partial^{2}}{\partial \theta_{l} \partial \theta_{j}}p_{i}(\theta).
\end{equation}
The second sum  in the right hand side is computed termwise as 
$$
 \frac{\partial}{\partial \theta_{l}} \left[\phi \left( \frac{p_{i}}{p_{i}(\theta)}  \right)  - \frac{p_{i}}{{p_{i}}(\theta) }   \phi^{'}\left( \frac{p_{i}}{{p_{i}}(\theta) }   \right)\right] =
$$
$$
=-\phi^{'}\left( \frac{p_{i}}{p_{i}(\theta)}  \right) \frac{p_{i}}{p^{2}_{i}(\theta)}  \frac{\partial}{\partial \theta_{l}}p_{i}(\theta)  +
 \frac{p_{i}}{p^{2}_{i}(\theta)}  \phi^{'}\left( \frac{p_{i}}{{p_{i}}(\theta) }   \right)  \frac{\partial}{\partial \theta_{l}}p_{i}(\theta) + 
$$
$$
 + \frac{p_{i}}{p_{i}(\theta) }   \phi^{''}\left( \frac{p_{i}}{p_{i}(\theta) } \right)  \frac{p_{i}}{p_{i}^{2}(\theta) }  \frac{\partial}{\partial \theta_{l}}p_{i}(\theta)
$$
\begin{equation}\label{jsddiv34}
=    \frac{p_{i}^{2}}{p_{i}^{3}(\theta) }   \phi^{''}\left( \frac{p_{i}}{p_{i}(\theta) } \right)  \frac{\partial}{\partial \theta_{l}}p_{i}(\theta).
\end{equation}
  By Equation~(\ref{jsddiv2}) 
$$
  \phi^{''}\left( \frac{p_{i}}{p_{i}(\theta) }\right)= \frac{p^{2}_{i}(\theta)}{2 p_{i}(p_{i}+p_{i}(\theta))}.
$$
 Hence 
\begin{equation}\label{jsddiv37}
 \frac{p_{i}^{2}}{p_{i}^{3}(\theta) }   \phi^{''}\left( \frac{p_{i}}{p_{i}(\theta) } \right) = \frac{p_{i}}{2p_{i}(\theta) }  \frac{1}{(p_{i}+p_{i}(\theta))}.  
\end{equation}
After  Equation~(\ref{jsddiv37}) has been  inserted in Equation~(\ref{jsddiv34}), Equation~(\ref{jsddiv34}) and Equation~(\ref{obslema}) are used in the rightmost expression in  Equation~(\ref{partialq2}) 
we obtain 
\begin{eqnarray}\label{hessele}
\frac{\partial^{2}}{\partial \theta_{l}\partial \theta_{j}}\Phi\left(  {\bf p},  \theta \right)&=&\sum_{i=1}^{k}   \phi^{'}\left( \frac{p_{i}(\theta)}  {p_{i}}  \right) \frac{\partial^{2}}{\partial \theta_{l} \partial \theta_{j}}p_{i}(\theta) \nonumber \\
& & \\
 & & + \sum_{i=1}^{k} \frac{p_{i}}{2p_{i}(\theta) }  \frac{1}{(p_{i}+p_{i}(\theta))} \frac{\partial}{\partial \theta_{l}}p_{i}(\theta) \frac{\partial}{\partial \theta_{j}}p_{i}(\theta). \nonumber
\end{eqnarray}

We re-organize the second term in the right-hand side of Equation~(\ref{hessele})  as 
$$
\sum_{i=1}^{k} \frac{p_{i}}{2p_{i}(\theta) }  \frac{1}{(p_{i}+p_{i}(\theta))} \frac{\partial}{\partial \theta_{l}}p_{i}(\theta) \frac{\partial}{\partial \theta_{j}}p_{i}(\theta) = \frac{1}{2}\sum_{i=1}^{k} \frac{\frac{\partial}{\partial \theta_{l}}p_{i}(\theta) \frac{\partial}{\partial \theta_{j}}p_{i}(\theta)}{p_{i}(\theta) }  \frac{p_{i} }{(p_{i}+p_{i}(\theta))}. 
$$
Next we use  $\frac{p_{i} }{(p_{i}+p_{i}(\theta))}  = 1- \frac{p_{i}(\theta) }{(p_{i}+p_{i}(\theta))}$ and  by Equation~(\ref{ijfisher3}) obtain  
\begin{equation}\label{fisinff}
=\frac{1}{2} I_{lj}(\theta)   -  \frac{1}{2}\sum_{i=1}^{k} 
\frac{ \frac{\partial}{\partial \theta_{l}}p_{i}(\theta) \frac{\partial}{\partial \theta_{j}}p_{i}(\theta)  }{(p_{i}+p_{i}(\theta))}. 
\end{equation}
Then the argument in the proof of  Lemma  \ref{infemma} and Equation~(\ref{ajsdfaktor})  can be repeated  to  verify   the   elementwise equality 
\begin{equation}\label{amarizel}
\sum_{i=1}^{k}\frac{ \frac{\partial}{\partial \theta_{l}}p_{i}(\theta) \ \frac{\partial}{\partial \theta_{j}}p_{i}(\theta)  }{(p_{i}+p_{i}(\theta))}= \left[A({\bf p}, \theta)^{T}A({\bf p}, \theta)\right]_{lj}.
\end{equation}
Hence we have obtained the sum  term in the right-hand side of Equation~(\ref{hessele})  as 
\begin{equation}\label{infsumma}
=\frac{1}{2} I_{lj}(\theta) - \frac{1}{2}\left[A({\bf p}, \theta)^{T}A({\bf p}, \theta)\right]_{lj}.
\end{equation}
When we use Equation~(\ref{infsumma}) in Equation~(\ref{hessele}), the asserted expression in Equation~(\ref{jsdhessian}) of Lemma \ref{jsdehesselem}  is there. 
 \end{proof}


\newpage
\vskip 0.2in
\bibliography{jsdmodv30.bib}

\end{document}